\documentclass[11pt, onecolumn]{IEEEtran}

\usepackage[numbers,sort&compress]{natbib}

\usepackage{amssymb}
\usepackage{amsmath,bm}
\usepackage{amssymb}
\usepackage{epsfig}
\usepackage{epsf}
\usepackage{subfigure}
\usepackage{graphicx}
\usepackage{algorithm}
\usepackage{algorithmic}
\usepackage{url}
\usepackage{amsthm}

\usepackage[bottom]{footmisc}

\date{}
\def\BibTeX{{\rm B\kern-.05em{\sc i\kern-.025em b}\kern-.08em
    T\kern-.1667em\lower.7ex\hbox{E}\kern-.125emX}}

\newtheorem{claim}{Claim}

\newtheorem{theorem}{Theorem}
\newtheorem{definition}{Definition}

\newtheorem{lemma}{Lemma}

\newtheorem{remark}{Remark}
\newtheorem{conjecture}{Conjecture}

\newcommand{\Text}[1]{\text{\textnormal{#1}}}

\newenvironment{lemmarep}[1]{\noindent {\bf #1.}\begin{it}}{\end{it}}

\author{Sina Lashgari, A. Salman Avestimehr, and Changho Suh 
\thanks{S. Lashgari is with the School of Electrical and Computer Engineering, Cornell University, Ithaca, NY (email: sl2232@cornell.edu, avestimehr@ece.cornell.edu); A. S. Avestimehr is with the EE department of University of Southern California, Los Angeles, CA 90089 (email: avestimehr@ee.usc.edu);
and C. Suh is with the Department of Electrical Engineering,
Korea Advanced Institute of Science and Technology, Daejeon, South Korea
(email: chsuh@kaist.ac.kr). The research of A. S. Avestimehr and S. Lashgari is supported  by NSF Grants CAREER 0953117, CCF-1161720, NETS-1161904, and ONR award N000141310094.

This work has been presented in part at the 51st Annual Allerton Conference on Communication, Control and Computing, 2013 ~\cite{Ours}.
}
}

\begin{document}

\title{Linear Degrees of Freedom of the $X$-Channel with Delayed CSIT}

\maketitle

\begin{abstract}
We establish the degrees of freedom of the two-user $X$-channel with delayed channel knowledge at transmitters (i.e., delayed CSIT), assuming linear coding strategies at the transmitters.  We derive a new upper bound and characterize the linear degrees of freedom of this network to be $\frac{6}{5}$. The converse builds upon our development of a general lemma that shows that, if two distributed transmitters employ linear strategies, the ratio of  the dimensions of received linear subspaces at the two receivers cannot exceed $\frac{3}{2}$, due to delayed CSIT. As a byproduct, we also apply this general lemma to the three-user interference channel with delayed CSIT, thereby deriving a new upper bound of $\frac{9}{7}$ on its linear degrees of freedom. This is the first bound that captures the impact of delayed CSIT on the degrees of freedom of this network, under the assumption of linear encoding strategies.

\end{abstract}

\section{Introduction}

The $X$-channel is a canonical setting for the information-theoretic study of interference management in wireless networks. This channel consists of two transmitters causing interference at two receivers, and each transmitter aims to communicate intended messages to both receivers. The question is: how can the transmitters optimally manage the interference and communicate their messages to the receivers? This problem has been studied extensively in the literature and various interference management techniques have been proposed. In particular, in~\cite{mimox}  it is shown that, quite surprisingly, one can significantly improve upon conventional interference management schemes (e.g., orthogonalization) and achieve $4/3$ degrees of freedom (DoF) by using \emph{interference alignment} (IA)~\cite{IA,alignment}.

However, in order to perfectly align the interference, the transmitters need to accurately know the  \emph{current} state of the channels, which is practically very challenging and may even be impossible (due to, for example, high mobility). Thus, a natural question would be: how can the transmitters optimally manage the interference with only \emph{delayed} knowledge of the channel state information (i.e., delayed CSIT)?

In the context of broadcast channel, Maddah-Ali and Tse in~\cite{MAT} have recently shown that delayed CSIT can still be very useful. In particular, for the multi-antenna broadcast channel with delayed CSIT, they developed an innovative transmission strategy that utilizes the past received signals to create signals of common interest to multiple receivers, hence significantly improving DoF by broadcasting them to the receivers. In a sense, these ``signals of common interest'' represent aligned interferences in the past receptions.

Subsequently in~\cite{retrospective,varanasi,abdoli,xchannel}, the impact of delayed CSIT has been explored for a variety of interference networks in which transmit antennas are now distributed at different locations. Unlike multi-antenna broadcast channels, in networks with distributed transmitters, it may not be possible for a transmitter to reconstruct previously received signals, since it may include other transmitters' signals that are not accessible to that transmitter. Hence, although interference alignment has happened in the past receptions, it may not be possible to construct the aligned interference locally at a transmitter and broadcast it to the receivers. Interestingly, even in this setting, delayed CSIT has shown to still provide DoF gains (see e.g., \cite{retrospective,varanasi, abdoli, xchannel}). In particular, for the $X$-channel, Ghasemi-Motahari-Khandani in~\cite{xchannel} developed a scheme that achieves DoF of $\frac{6}{5}$ with delayed CSIT, which is strictly larger than its DoF with no-CSIT (i.e., 1 DoF). However, given that the only upper bound on the DoF of this network is the one with instantaneous CSIT (i.e., $\frac{4}{3}$ DoF), it remains still open whether $\frac{6}{5}$ is the fundamental limit on the DoF of $X$-channel with delayed CSIT, or whether there are more efficient interference management techniques.

Our main contribution in this work is to show that the DoF of the Gaussian $X$- channel with delayed CSIT is indeed $\frac{6}{5}$, under the assumption that only linear encoding schemes are employed at the transmitters. Under this constraint,  only a linear combination of information symbols are allowed to be transmitted at each time. In fact, all of the interference management strategies with delayed CSIT that are developed thus far (e.g.,~\cite{MAT, retrospective, varanasi, abdoli, xchannel}) fall into this category.

The key part of the converse is the development of a general lemma, namely ``Rank Ratio Inequality", that bounds the maximum ratio of the dimensions of received linear-subspaces (at the two receivers) that are created by \emph{distributed} transmitters with delayed CSIT. More specifically, we show that if two distributed transmitters with delayed CSIT employ linear strategies, the ratio of the dimensions of the received signals  cannot exceed $\frac{3}{2}$. With instantaneous CSIT, this ratio can be as large as $2$, and with no CSIT, this ratio is always $1$. As a result, this lemma captures the fundamental impact of delayed CSIT on the dimension of received subspaces. Also, in the case of two centralized transmitters (e.g., multi-antenna BC), this ratio can be as large as $2$, therefore Rank Ratio Inequality also captures the impact of \emph{distributed transmitters} on the dimension of received subspaces. Rank Ratio Inequality can also be viewed  as a generalization of the ``entropy leakage Lemma'' in \cite{VMABinaryIC, vahid2013communication}, which considers a broadcast channel with binary fading, and bounds the maximum ratio of the entropy of received signals at two different receivers.

We also demonstrate how our lemma can be applied to any arbitrary network, in which a receiver decodes its desired message in the presence of two interferers. As an example, we apply the lemma to the three-user interference channel with delayed CSIT and derive a new upper bound of $\frac{9}{7}$ on its linear DoF. This is the first upper bound that captures the impact of delayed CSIT on the degrees of freedom of this network.

{\bf Other Related Results.} 
There have been several converse techniques developed in the literature for networks with delayed CSIT.
For the MISO broadcast channel with delayed CSIT, Maddah-Ali and Tse ~\cite{MAT} have provided an upper bound  based on the genie-aided bounding technique. This technique essentially consists of two steps. First, signals of a set of receivers are given to other set of receivers such that the enhanced network becomes a physically degraded broadcast channel. Using the fact that feedback cannot increase capacity for physically degraded broadcast channels~\cite{ElGamal}, we can then take the non-feedback upper bound as that of the original feedback channel. This technique has also been used in~\cite{MISOBCAllerton} to approximate the capacity of MISO broadcast channel with delayed CSIT, and in \cite{Tassiulas} in the context of broadcast erasure channels with feedback. Also, for time correlated MISO broadcast channel with delayed CSIT, a converse has been proposed in~\cite{Yang}, where the essential element of the converse is the use of the extremal inequality~\cite{Extremal} that bounds the weighted difference of differential entropies at the receivers. Finally, for MIMO interference channel with delayed CSIT,  a converse has been proposed in~\cite{varanasi}, which utilizes the fact that for delayed CSIT, the signals received at different receivers in a timeslot are statistically equivalent; therefore, the entropy of received signals at different receivers in a certain timeslot are equal when conditioned on  past received signals at any specific receiver. 


{\bf Notation.} We use small letters for scalars, arrowed letters (e.g. $\vec x$) for vectors,  capital letters for matrices, and a calligraphic font for sets.
Furthermore, we use bold letters for random entities, and non-bold letters for deterministic values (e.g., realizations of random variables).

\section{System Model \& Main Results} \label{model}

We consider the Gaussian $X$-channel depicted in Fig.~\ref{xch}. It consists of
 two transmitters and two receivers, and each transmitter has a separate message for each of the receivers.
Each node is equipped with a single antenna.

\begin{figure}[h!]
\centering
\includegraphics[scale=.3, trim= 10mm 10mm 10mm 10mm]{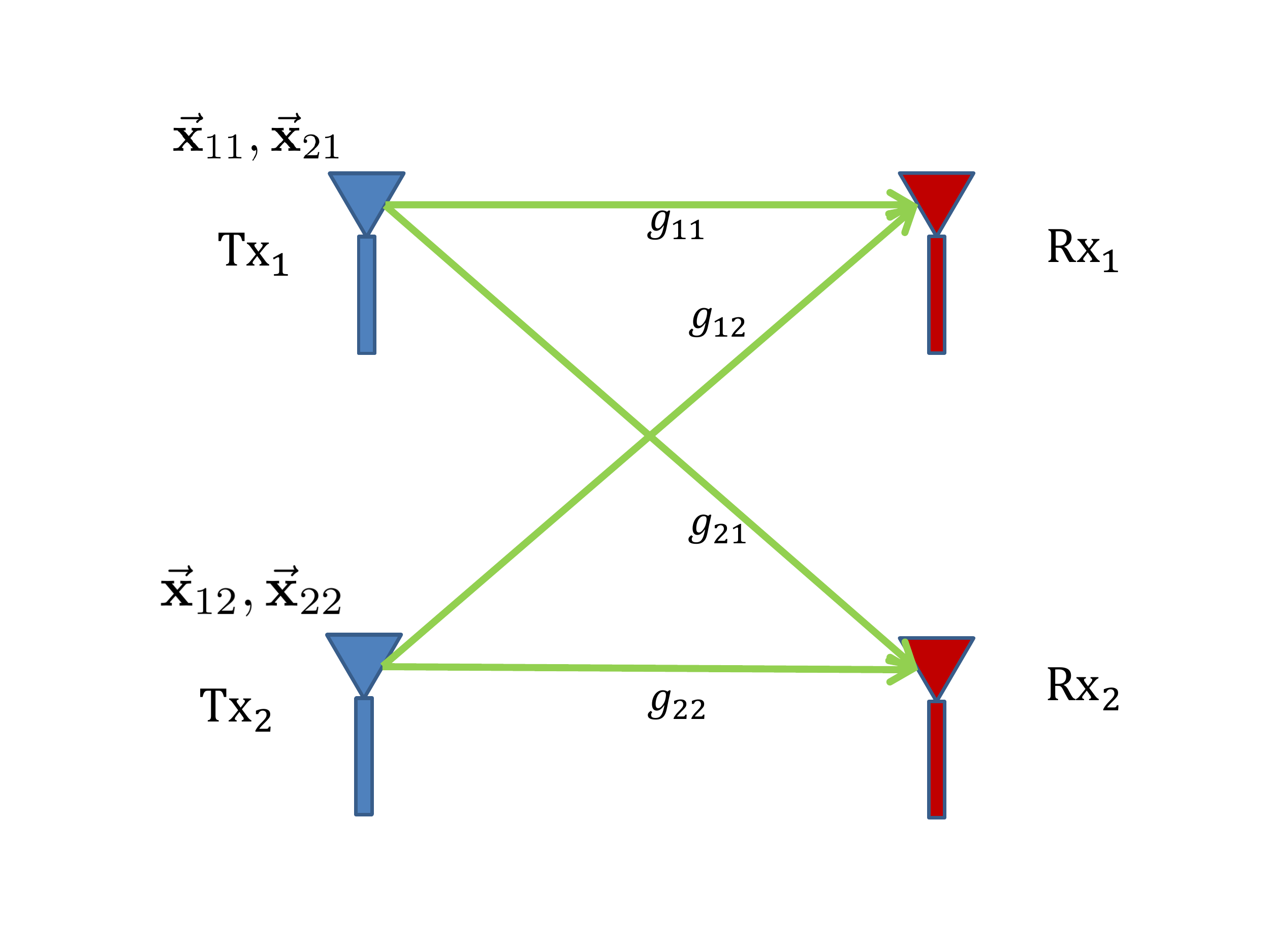}\\
\caption{Network configuration for X-channel.
There are two transmitters and two receivers, where each transmitter has a message for each receiver. We assume time-varying channels, with delayed CSIT.
}\label{xch}
\end{figure}

The received signal at $\text{Rx}_{k}$ ($k\in \{1,2\}$) at time $t$ is given by
\begin{equation}
\bold{y}_{k}(t)=\bold{g}_{k1}(t)\bold{x}_{1}(t)+ \bold{g}_{k2}(t)\bold{x}_{2}(t)+\bold{z}_{k}(t),
\end{equation}
where
$\bold{x}_{j}(t)$ is the transmit signal of $\text{Tx}_j$;
$\bold{g}_{kj}(t) \in \mathbb C$ indicates a channel from 
$\text{Tx}_{j}$ to $\text{Rx}_{k}$;
and $\bold{z}_{k}(t) \sim \mathcal C\mathcal N(0,1)$.
The channel coefficients of $\bold{g}_{kj}(t)$'s are i.i.d across time and users, and they are drawn from a continuous distribution.
We denote by $\bm{ \mathcal {G}}(t)$ the set of all four channel coefficients at time $t$.
In addition, we denote by $\bm{\mathcal {G}}^n$ the set of all channel coefficients from time 1 to $n$, i.e., $$\bm{\mathcal {G}}^n=\{\bold{g}_{kj}(t): k,j \in \{1,2\}, t=1,\ldots,n\}.$$

Denoting the vector of transmit signals for $\text{Tx}_j$ in a block of length $n$ by $\vec {\bf x}_j^n$,
each transmitter $\text{Tx}_j$ obeys an average power constraint, $\frac{1}{n}E\{|| \vec {\bf x}_{j}^n||^2 \}\leq P$.
We assume  delayed channel state information at the transmitters (CSIT).
In other words, at  time $t$, only the states of the past $\bm{\mathcal {G}}^{t-1}$ are known to the transmitters.
Furthermore, we assume that receivers have instantaneous CSIT, meaning that at time $t$, $\bm{\mathcal {G}}^{t}$ is known to all receivers.

We restrict ourselves to linear coding strategies as defined in~\cite{bresler}, in which DoF simply represents the dimension of the linear subspace of transmitted signals.
More specifically, consider a communication scheme with block length $n$, in which transmitter $\text{Tx}_{j}$ wishes to transmit a vector $\vec {\bf x}_{kj} \in \mathbb C^{m_{kj}(n)}$ of $m_{kj}(n) \in \mathbb{N}$ information symbols to $\text{Rx}_k$ ($j,k \in \{1,2\}$).
These information symbols are then modulated with precoding vectors $\vec{\bf v}_{kj}(t)\in \mathbb C^{ m_{kj}(n)}$ at times $t=1,2,\ldots,n$. Note that the precoding vector $\vec{\bf v}_{kj}(t)$ depends only upon the outcome of $\bm{\mathcal {G}}^{t-1}$ due to the delayed CSIT constraint:
\begin{equation}  \label{coding}
\vec{ v}_{kj}(t)=f_{k,j,t}^{(n)} \left( \mathcal G^{t-1} \right).
\end{equation}
Based on this linear precoding, $\text{Tx}_j$ will then send $ {\bf x}_{j}(t)=\vec{\bf v}_{1j}(t)^\top \vec {\bf x}_{1j}+ \vec{\bf v}_{2j}(t)^\top \vec {\bf x}_{2j}$ at time $t$.
We denote by $\bold{V}_{kj}^n\in \mathbb C^{n\times m_{kj}(n)}$ the overall precoding matrix of $\text{Tx}_j$ for $\text{Rx}_k$, such that   the $t$-th row of $\bold{V}_{kj}^n$ is $\vec{\bf v}_{kj}(t)^\top$).
In addition, we denote the precoding functions used by $\text{Tx}_j$ by $f_j^{(n)}=  \{ f_{1,j,t}^{(n)}  ,  f_{2,j,t}^{(n)}    \}_{t=1}^{n}$, $j=1,2$.

Based on the above setting, the received signal at $\text{Rx}_{k}$ ($k\in \{1,2\}$) after the $n$ time steps of the communication will be
\begin{equation}
\vec {\bf y}_{k}^n= \bold{G}_{k1}^n (\bold{V}_{11}^n \vec {\bf x}_{11} + \bold{V}_{21}^n \vec {\bf x}_{21}) + \bold{G}_{k2}^n (\bold{V}_{12}^n \vec {\bf x}_{12} + \bold{V}_{22}^n \vec {\bf x}_{22}) + \vec {\bf z}_{k}^n,
\end{equation}
 where $\bold{G}_{kj}^n$ is the $n\times n$ diagonal matrix whose $t$-th element on the diagonal is $\bold{g}_{kj}(t)$.
\footnote{For $j,k\in \{1,2\}$, we define $\bold{G}_{kj}^0 \bold{V}_{kj}^0  \triangleq 0_{1\times m_{kj}(n)} $; therefore, for instance, we have $\Text{rank}\left[ \bold{G}_{k1}^0 \bold{V}_{k1}^0 \quad  \bold{G}_{k2}^0\bold{V}_{k2}^0 \right]=0$, $k\in \{1,2\}$.} 
Now, consider the decoding of $\vec {\bf x}_{kj}$ at $\text{Rx}_k$ (i.e., the $m_{kj}(n)$ information symbols of $\text{Tx}_{j}$ for $\text{Rx}_k$). The corresponding interference subspace at $\text{Rx}_k$ will be
 \begin{align*}
   \bm{\mathcal I}_{kj}=\text{colspan} \left( \bold{G}_{kj}^n \bold{V}_{k'j}^n \right)  \cup\text{colspan} \left( \bold{G}_{kj'}^n \bold{V}_{kj'}^n \right)  \cup \text{colspan} \left( \bold{G}_{kj'}^n \bold{V}_{k'j'}^n \right)
\end{align*} 
where $j'=3-j, k'=3-k$, and $\text{colspan}(.)$ of a matrix corresponds to the sub-space that is spanned by its columns. For instance, $\bm{\mathcal I}_{11}=\text{colspan} (\bold{G}_{11}^n \bold{V}_{21}^n) \cup \text{colspan} (\bold{G}_{12}^n \bold{V}_{12}^n) \cup  \text{colspan} (\bold{G}_{12}^n \bold{V}_{22}^n) $. Let $\bm{\mathcal I}_{kj}^c \subseteq \mathbb{C}^n$ denote the subspace orthogonal to $\bm{\mathcal I}_{kj}$. Then, in the regime of asymptotically high transmit powers (i.e., ignoring the noise), the decodability of information symbols from $\text{Tx}_{j}$ at $\text{Rx}_k$ corresponds to the constraints that  the image of $\text{colspan}(\bold{G}_{kj}^n \bold{V}_{kj}^n)$ on $\bm{\mathcal I}_{kj}^c$ has dimension $m_{kj}(n)$:
\begin{equation}
\Text{dim} \left( \text{Proj}_{\bm{\mathcal I}_{kj}^c} \text{colspan} \left( \bold{G}_{kj}^n \bold{V}_{kj}^n \right) \right) =\Text{dim} \left( \text{colspan} \left( \bold{V}_{kj}^n \right) \right) = m_{kj}(n).\label{decode}
\end{equation}

Based on this setting, we now define the sum linear degrees of freedom of the $X$-channel.
\begin{definition}
Four-tuple $(d_{11},d_{12},d_{21},d_{22})$ degrees of freedom are linearly achievable if there exists a sequence
\\$\{  f_1^{(n)},f_2^{(n)} \}_{n=1}^{\infty}$ such that for each $n$ and the choice of $(m_{11}(n),m_{12}(n),m_{21}(n),m_{22}(n))$, $(\bold{V}_{11}^n,\bold{V}_{12}^n,\bold{V}_{21}^n,\bold{V}_{22}^n)$ satisfy the decodability condition of (\ref{decode}) with probability 1, and $\forall (j,k)$,
\begin{equation}
d_{kj}=\lim_{n\to\infty} \frac{m_{kj}(n)}{n}.
\end{equation}

We also define the linear degrees of freedom region $\mathcal D$ as the closure of the set of all achievable 4-tuples $(d_{11},d_{12},d_{21},d_{22})$.
Furthermore, the sum linear degrees of freedom ($\Text{DoF}_{\Text{L-sum}}$) is then defined as follows:
\begin{equation}
\Text{DoF}_{\Text{L-sum}}=\max\sum_{k,j\in \{1,2\}}^{}d_{kj},\qquad \textrm{s.t. }\quad (d_{11},d_{12},d_{21},d_{22})\in\mathcal D.
\end{equation}
\end{definition}

In case transmitters have instantaneous CSIT,
it was shown in \cite{IA,xjafar} that the sum degrees of freedom is $\frac{4}{3}$. The achievability uses interference alignment that enables us to deliver four symbols over three timeslots.
On the other hand, in the non-CSIT case, one can readily see that the received signals at the two receivers are statistically identical and therefore the DoF collapses to 1, which is that of the multiple access channel.
For the case of delayed CSIT, Ghasemi-Motahari-Khandani in \cite{xchannel} develops a new scheme that achieves the sum DoF of $\frac{6}{5}$. 

Our main result in this paper is the following theorem, proved in Section~\ref{convers}, which states that $\frac{6}{5}$ is the maximum DoF that can be achieved using linear encoding schemes.

\begin{theorem}\label{theorem1}
For the X-channel with delayed CSIT,
\begin{equation}
\Text{DoF}_{\Text{L-sum}}= \frac{6}{5}.
\end{equation}
\end{theorem}

Our converse proof builds upon the following key lemma, which is proved in Section~\ref{sec:keyLemmaProof}.

\begin{lemma}
\label{lemma1}
{\bf (Rank Ratio Inequality)}
For any linear coding strategy $\{  f_1^{(n)},f_2^{(n)} \}$, with corresponding $ \bold{V}_{11}^{n},\bold{V}_{12}^{n} $ as defined in (\ref{coding}),
\begin{align}
\label{ineq1}
\Text{rank}\left[ \bold{G}_{11}^n \bold{V}_{11}^n \quad  \bold{G}_{12}^n\bold{V}_{12}^n \right] \stackrel{a.s.}{\leq}  \frac{3}{2} \Text{rank} \left[\bold{G}_{21}^n \bold{V}_{11}^n \quad  \bold{G}_{22}^n\bold{V}_{12}^n \right].
\end{align}
\end{lemma}

\begin{remark}
Note that this lemma holds for any arbitrary network (or sub-network) with two transmitters and two receivers. It does not require any specific decodability assumption at receivers. The inequality of~\eqref{ineq1} says that the ratio of the ranks of received beamforming matrices at $\Text{Rx}_1$  and  $\Text{Rx}_2$ is at most $\frac{3}{2}$.
For the case of having instantaneous CSIT, one can show that the ratio of $\Text{rank}[\bold{G}_{11}^n \bold{V}_{11}^n \quad  \bold{G}_{12}^n\bold{V}_{12}^n]$ to
$\Text{rank}[\bold{G}_{21}^n \bold{V}_{11}^n \quad  \bold{G}_{22}^n\bold{V}_{12}^n]$ can be up to 2.\footnote{To see this, consider the following two-timeslot scheme. In time 1, $\text{Tx}_1, \text{Tx}_2$ send $\bold{x}_1, \bold{x}_2$ respectively. $\text{Rx}_2$ then gets $\bold{g}_{21}(1)\bold{x}_1+ \bold{g}_{22}(1)\bold{x}_2$.
In time 2, $\text{Tx}_1, \text{Tx}_2$ send $\frac{\bold{g}_{21}(1)}{\bold{g}_{21}(2)} \bold{x}_1, \frac{\bold{g}_{22}(1)}{\bold{g}_{22}(2)} \bold{x}_2$ respectively. $\text{Rx}_2$ then gets the same equation as the one received in time 1. On the other hand, $\text{Rx}_1$ gets a new equation almost surely. Therefore, the rank of the received signal at $\text{Rx}_1$ can be twice that of $\text{Rx}_2$. Also one can readily show that the two is  the maximum that can be achieved. 
}
Hence, Lemma~\ref{lemma1} characterizes the impact of delayed CSIT on the maximum ratio of the ranks of received beamforming matrices. 
\end{remark}
\begin{remark}
Lemma~\ref{lemma1} can be viewed as a generalization of  the ``entropy leakage Lemma'' in \cite{VMABinaryIC}. Entropy leakage lemma in \cite{VMABinaryIC}  considers a broadcast channel with binary fading, and bounds the maximum ratio of the entropy of received signals at two different receivers. In fact, Rank Ratio Inequality can be viewed as an extension of this lemma to the case of two distributed transmitters with linear encoding strategies, in which the entropy is approximated by the rank of the received beamforming matrices.
\end{remark}

\section{Proof of Theorem \ref{theorem1}} \label{convers}

In this section we will prove Theorem~\ref{theorem1}.

\subsection{Achievability}
As mentioned in the previous section, the achievability is provided in \cite{xchannel}, and utilizes a linear encoding scheme to achieve $\frac{6}{5}$.
 Here we review the scheme to illustrate how beamforming vectors are chosen.
We set $n=5, m_{11}(n)=2 , m_{12}(n)=1 , m_{21}(n)=1 , m_{22}(n)=2$. Let the information symbols of the transmitters be denoted by
\begin{equation}
\vec {\bf x}_{11}=\left[\begin{array}{c}a_1 \\a_2 \end{array}\right], \quad \vec {\bf x}_{12}=\left[\begin{array}{c}b_1\end{array}\right], \quad \vec {\bf x}_{21}=\left[\begin{array}{c}c_1  \end{array}\right], \quad \vec {\bf x}_{22}=\left[\begin{array}{c}d_1 \\d_2 \end{array}\right].
\end{equation}
In $t=1$, $\text{Tx}_1$ sends $a_1$, and  $\text{Tx}_2$ sends $b_1$, which corresponds to choosing the following beamforming vectors at the transmitters
$$\vec { v}_{11}=\left[\begin{array}{c}1 \\0 \end{array}\right], \quad \vec { v}_{12}=\left[\begin{array}{c}1\end{array}\right], \quad \vec { v}_{21}=\left[\begin{array}{c}0  \end{array}\right], \quad \vec { v}_{22}=\left[\begin{array}{c}0 \\0 \end{array}\right].$$
In $t=2$, $\text{Tx}_1$ sends $a_2$, and  $\text{Tx}_2$ sends $b_1$, which corresponds to choosing the following beamforming vectors at the transmitters
$$\vec { v}_{11}=\left[\begin{array}{c}0 \\1 \end{array}\right], \quad \vec { v}_{12}=\left[\begin{array}{c}1\end{array}\right], \quad \vec { v}_{21}=\left[\begin{array}{c}0  \end{array}\right], \quad \vec { v}_{22}=\left[\begin{array}{c}0 \\0 \end{array}\right].$$
Therefore, by the end of $t=2$, $\text{Rx}_2$ can cancel $b_1$ from its received signals to recover an equation only involving $a_1$ and $a_2$, denoted by $\vec {\bf m}_1^\top \vec {\bf x}_{11}$. It is easy to see that, if this equation is delivered to $\text{Rx}_1$, it can decode all of its desired information symbols (i.e., $\vec {\bf x}_{11}$ and  $\vec {\bf x}_{12}$). Hence, it is an equation of interest to $\text{Rx}_1$ that is known at $\text{Rx}_2$, and can be created by $\text{Tx}_1$. 

A similar schemes is applied in the next two time steps. More specifically, in $t=3$, $\text{Tx}_1$ sends $c_1$, and  $\text{Tx}_2$ sends $d_1$, which corresponds to choosing the following beamforming vectors at the transmitters
$$\vec { v}_{11}=\left[\begin{array}{c}0 \\0 \end{array}\right], \quad \vec { v}_{12}=\left[\begin{array}{c}0\end{array}\right], \quad \vec { v}_{21}=\left[\begin{array}{c}1  \end{array}\right], \quad \vec { v}_{22}=\left[\begin{array}{c}1 \\0 \end{array}\right].$$
In $t=4$, $\text{Tx}_1$ sends $c_1$, and  $\text{Tx}_2$ sends $d_2$, which corresponds to choosing the following beamforming vectors at the transmitters
$$\vec { v}_{11}=\left[\begin{array}{c}0 \\0 \end{array}\right], \quad \vec { v}_{12}=\left[\begin{array}{c}0\end{array}\right], \quad \vec { v}_{21}=\left[\begin{array}{c}1  \end{array}\right], \quad \vec { v}_{22}=\left[\begin{array}{c}0 \\1 \end{array}\right].$$
Therefore, by the end of $t=4$, $\text{Rx}_1$ can cancel $c_1$ from its received signals to recover an equation only involving $d_1$ and $d_2$, denoted by $\vec {\bf m}_2^\top \vec {\bf x}_{22}$. Again, it is easy to see that, if this equation is delivered to $\text{Rx}_2$, it can decode all of its desired information symbols (i.e., $\vec {\bf x}_{21}$ and  $\vec {\bf x}_{22}$). Hence, it is an equation of interest to $\text{Rx}_2$ that is known at $\text{Rx}_1$, and can be created by $\text{Tx}_2$.\footnote{One can check that $\vec {\bf m}_1=[\bold{g}_{22}(2)\bold{g}_{21}(1) \quad   - \bold{g}_{22}(1)\bold{g}_{21}(2)]^\top    $, and $\vec {\bf m}_2=  [ \bold{g}_{12}(3)\bold{g}_{11}(4)  \quad    -  \bold{g}_{11}(3) \bold{g}_{12}(4) ]^\top  $.}

Now, in $t=5$, $\text{Tx}_1$ sends $\vec {\bf m}_1^\top \vec {\bf x}_{11}$, and  $\text{Tx}_2$ sends $\vec {\bf m}_2^\top \vec {\bf x}_{22}$.
Since each of these transmit signals is already known at one of the receivers, after this transmission, $\text{Rx}_1$ will  recover $\vec {\bf m}_1^\top \vec {\bf x}_{11}$ and $\text{Rx}_2$ will  recover $\vec {\bf m}_2^\top \vec {\bf x}_{22}$. Therefore, all information symbols are delivered to their corresponding receivers, achieving sum DoF of $\frac{6}{5}$.

\subsection{Converse}
We will now prove the converse, which is the main contribution of the paper. As mentioned in Section \ref{model}, the key idea behind the converse is Lemma~\ref{lemma1}, which we restate below (proof of Lemma \ref{lemma1} is provided in Section  \ref{sec:keyLemmaProof}).

\begin{lemmarep}{Lemma~\ref{lemma1}}
{\bf (Rank Ratio Inequality)}
For any linear coding strategy $\{  f_1^{(n)},f_2^{(n)} \}$, with corresponding $ \bold{V}_{11}^{n},\bold{V}_{12}^{n} $ as defined in (\ref{coding}),
\begin{align}
 \Text{rank}\left[ \bold{G}_{11}^n \bold{V}_{11}^n \quad  \bold{G}_{12}^n\bold{V}_{12}^n \right] \stackrel{a.s.}{\leq}  \frac{3}{2} \Text{rank} \left[\bold{G}_{21}^n \bold{V}_{11}^n \quad  \bold{G}_{22}^n\bold{V}_{12}^n \right].
\end{align}
\end{lemmarep}

To prove the converse we also need the following three lemmas.
The following lemma states the sub-modularity property of rank of matrices (see ~\cite{Lovasz} for more details).

\begin{lemma} \label{ranksubmod}
{\bf (Sub-modularity of rank)} Consider a matrix $A^{m\times n} \in \mathbb C^{m\times n}$.
Let $A_I$, $I\subseteq \{  1,2 , \ldots , n\}$ denote the sub-matrix of $A$ created by those columns in $A$ which have their indices in $I$.
Then, for any $I_1,I_2 \subseteq \{  1,2 , \ldots , n\}$ we have
\begin{equation}
\Text{rank}[A_{I_1}] + \Text{rank}[A_{I_2}] \geq \Text{rank}[A_{I_1\cap I_2}] + \Text{rank}[A_{I_1\cup I_2}].
\end{equation}
\end{lemma}

The following lemma is helpful in providing an equivalent condition for decodability of messages in (\ref{decode}), whose proof is based on basic linear algebra and omitted.

\begin{lemma} \label{eq}
For two matrices $A,B$ of the same row size, 
\begin{equation}
\Text{dim} (\Text{Proj}_{\Text{colspan}(B)^c}\Text{colspan}(A)) = \Text{rank}[A \quad B] - \Text{rank}[B],
\end{equation}
where $\Text{Proj}_{\Text{colspan}(B)^c}\Text{colspan}(A)$ is the orthogonal projection of column span of $A$ on the orthogonal complement of the column span of $B$.
\end{lemma}



Finally, the following lemma, whose proof is based on the sub-modularity of the rank function (Lemma \ref{ranksubmod}), will be useful later in the converse proof.

\begin{lemma} \label{cl1}
Suppose that for four matrices ${A}, {B}, {C}, {D}$ with the same number of rows,
\begin{align}
&    \Text{rank}[{A}]+ \Text{rank}[{B}\quad {C} \quad {D}] =  \Text{rank}[{A} \quad {B}\quad {C} \quad {D}], \nonumber \\
&    \Text{rank}[{B}]+ \Text{rank}[{A}\quad {C} \quad {D}] =  \Text{rank}[{A} \quad {B}\quad {C} \quad {D}]. \label{eq:need1}
\end{align}
Then,
$$ \Text{rank}[{A}] + \Text{rank}[  {B}]+ \Text{rank}[{C} \quad {D}] =  \Text{rank}[{A} \quad {B}\quad {C} \quad {D}]. $$
\end{lemma}

\begin{proof}
Note that $ \Text{rank}[{A}] + \Text{rank}[  {B}]+ \Text{rank}[{C} \quad {D}] \geq  \Text{rank}[{A} \quad {B}\quad {C} \quad {D}]$. Hence, in order to prove Lemma ~\ref{cl1}, we only need to prove the inequality in the other direction. Now, according to the assumptions in the Lemma, and using sub-modularity of the rank (Lemma \ref{ranksubmod}), we have
\begin{align}
\Text{rank}[{A}]+\Text{rank}[{B}] \stackrel{(\ref{eq:need1})}{=} &  \Text{rank}[{A} \quad {B}\quad {C} \quad {D}]- \Text{rank}[{B}\quad {C} \quad {D}]     +   \Text{rank}[{A} \quad {B}\quad {C} \quad {D}] -  \Text{rank}[{A}\quad {C} \quad {D}] \nonumber \\
 \stackrel{(\text{sub-modularity})}{\leq} & \Text{rank}[{A} \quad {B}\quad {C} \quad {D}]- \Text{rank}[{B}\quad {C} \quad {D}]    +   \Text{rank}[{B}\quad {C} \quad {D}] -  \Text{rank}[{C} \quad {D}] \nonumber \\
\nonumber =&  \Text{rank}[{A} \quad {B}\quad {C} \quad {D}]- \Text{rank}[ {C} \quad {D}].
\end{align}
\end{proof}

We are now ready to prove the converse.
In particular, we prove the following two inequalities:
\begin{align}
& (d_{11}+d_{12})+\frac{3}{2}(d_{21}+d_{22})\leq \frac{3}{2}\label{2r2x}\\
&\frac{3}{2}(d_{11}+d_{12})+(d_{21}+d_{22})\leq \frac{3}{2}.\label{2r1x}
\end{align}
The desired result follows from summing the above two inequalities.
By symmetry, we only need to prove (\ref{2r2x}).
Suppose $(d_{11},d_{12},d_{21},d_{22})\in \mathcal D$, i.e., there exists a sequence
$\{ f_1^{(n)},f_2^{(n)}  \}_{n=1}^{\infty}$ resulting in linearly achieving
 $\{ m_{11}(n),m_{12}(n),m_{21}(n),m_{22}(n)  \}_{n=1}^{\infty}$ with probability 1, and $d_{kj}=\lim_{n\to\infty} \frac{m_{kj}(n)}{n}$.
First, note that 
\begin{equation}  \label{equil}
\Text{dim} \left( \text{colspan} (\bold{V}_{kj}^n) \right) \stackrel{a.s.}{=} \Text{dim} \left(\text{colspan} (\bold{G}_{kj}^n\bold{V}_{kj}^n) \right),
\end{equation}
 due to the continuous distribution of $\bold{g}_{kj}(t)$ for any $t$.
Therefore, by~(\ref{equil}) and Lemma~\ref{eq}, we conclude that if (\ref{decode}) occurs with probability 1, then for
$j,k\in\{1,2\}$
and $j'=3-j, k'=3-k,$
\begin{equation}  \label{altcond}
 \Text{rank}[\bold{G}_{kj}^n\bold{V}_{k'j}^n   \quad \bold{G}_{kj'}^n\bold{V}_{kj'}^n  \quad \bold{G}_{kj'}^n\bold{V}_{k'j'}^n         ]+\Text{rank}[\bold{G}_{kj}^n\bold{V}_{kj}^n] \stackrel{a.s.}{=}\Text{rank}[\bold{G}_{k1}^n\bold{V}_{k1}^n \quad \bold{G}_{k2}^n\bold{V}_{k2}^n \quad \bold{G}_{k1}^n\bold{V}_{k'1}^n \quad \bold{G}_{k2}^n\bold{V}_{k'2}^n],
\end{equation}
Thus, we consider (\ref{altcond}) as the equivalent decodability condition, which consists of  the following four equations:
\begin{align}
& \Text{rank}[\bold{G}_{11}^n \bold{V}_{11}^n] + \Text{rank}[\bold{G}_{11}^n\bold{V}_{21}^n \quad \bold{G}_{12}^n\bold{V}_{12}^n \quad \bold{G}_{12}^n\bold{V}_{22}^n] \stackrel{a.s.}{=} \Text{rank}[\bold{G}_{11}^n\bold{V}_{11}^n \quad \bold{G}_{11}^n\bold{V}_{21}^n \quad \bold{G}_{12}^n\bold{V}_{12}^n \quad \bold{G}_{12}^n\bold{V}_{22}^n]  \label{condi1}  \\
& \Text{rank}[\bold{G}_{12}^n \bold{V}_{12}^n] + \Text{rank}[\bold{G}_{11}^n\bold{V}_{11}^n \quad \bold{G}_{11}^n\bold{V}_{21}^n  \quad \bold{G}_{12}^n\bold{V}_{22}^n]\stackrel{a.s.}{=} \Text{rank}[\bold{G}_{11}^n\bold{V}_{11}^n \quad \bold{G}_{11}^n\bold{V}_{21}^n \quad \bold{G}_{12}^n\bold{V}_{12}^n \quad \bold{G}_{12}^n\bold{V}_{22}^n] \label{condi2}   \\
& \Text{rank}[\bold{G}_{21}^n \bold{V}_{21}^n] + \Text{rank}[\bold{G}_{21}^n\bold{V}_{11}^n  \quad \bold{G}_{22}^n\bold{V}_{12}^n \quad \bold{G}_{22}^n\bold{V}_{22}^n] \stackrel{a.s.}{=} \Text{rank}[\bold{G}_{21}^n\bold{V}_{11}^n \quad \bold{G}_{21}^n\bold{V}_{21}^n \quad \bold{G}_{22}^n\bold{V}_{12}^n \quad \bold{G}_{22}^n\bold{V}_{22}^n]  \label{condi3}    \\
& \Text{rank}[\bold{G}_{22}^n \bold{V}_{22}^n] + \Text{rank}[ \bold{G}_{21}^n\bold{V}_{11}^n \quad \bold{G}_{21}^n\bold{V}_{21}^n \quad \bold{G}_{22}^n\bold{V}_{12}^n] \stackrel{a.s.}{=} \Text{rank}[\bold{G}_{21}^n\bold{V}_{11}^n \quad \bold{G}_{21}^n\bold{V}_{21}^n \quad \bold{G}_{22}^n\bold{V}_{12}^n \quad \bold{G}_{22}^n\bold{V}_{22}^n]. \label{condi4}
\end{align}

Hence, by (\ref{condi1}), (\ref{condi2}), and Lemma \ref{cl1},
\begin{equation}  \label{joint1}
\Text{rank}[\bold{G}_{11}^n \bold{V}_{11}^n] + \Text{rank}[\bold{G}_{12}^n \bold{V}_{12}^n] \stackrel{a.s.}{=} \Text{rank}[\bold{G}_{11}^n\bold{V}_{11}^n \quad \bold{G}_{11}^n\bold{V}_{21}^n \quad \bold{G}_{12}^n\bold{V}_{12}^n \quad \bold{G}_{12}^n\bold{V}_{22}^n] -\Text{rank}[\bold{G}_{11}^n\bold{V}_{21}^n \quad \bold{G}_{12}^n\bold{V}_{22}^n].
\end{equation}
In addition, by (\ref{condi3}), (\ref{condi4}), and Lemma \ref{cl1},
\begin{equation} \label{joint2}
\Text{rank}[\bold{G}_{21}^n \bold{V}_{21}^n] + \Text{rank}[\bold{G}_{22}^n \bold{V}_{22}^n] \stackrel{a.s.}{=}  \Text{rank}[\bold{G}_{21}^n\bold{V}_{11}^n \quad \bold{G}_{21}^n\bold{V}_{21}^n \quad \bold{G}_{22}^n\bold{V}_{12}^n \quad \bold{G}_{22}^n\bold{V}_{22}^n] -\Text{rank}[\bold{G}_{21}^n\bold{V}_{11}^n \quad \bold{G}_{22}^n\bold{V}_{12}^n].
\end{equation}

Therefore, we have
\begin{align}
m_{11}(n)+m_{12}(n)&+\frac{3}{2}(m_{21}(n)+m_{22}(n)) \nonumber\\
\stackrel{a.s.}{=} &\Text{rank}[\bold{V}_{11}^n] + \Text{rank}[\bold{V}_{12}^n] +\frac{3}{2} (\Text{rank}[\bold{V}_{21}^n] + \Text{rank}[\bold{V}_{22}^n]) \nonumber\\
\stackrel{a.s.}{=} &\Text{rank}[\bold{G}_{11}^n \bold{V}_{11}^n] + \Text{rank}[\bold{G}_{12}^n \bold{V}_{12}^n] +\frac{3}{2} (\Text{rank}[ \bold{G}_{21}^n  \bold{V}_{21}^n] + \Text{rank}[  \bold{G}_{22}^n  \bold{V}_{22}^n]) \nonumber\\
\substack{ \text{(\ref{joint1}), (\ref{joint2})}   \\ a.s.\\=} &  \Text{rank}[\bold{G}_{11}^n\bold{V}_{11}^n \quad \bold{G}_{12}^n\bold{V}_{12}^n \quad \bold{G}_{11}^n\bold{V}_{21}^n \quad \bold{G}_{12}^n\bold{V}_{22}^n] -\Text{rank}[\bold{G}_{11}^n\bold{V}_{21}^n \quad \bold{G}_{12}^n\bold{V}_{22}^n]  \nonumber\\
&+ \frac{3}{2} ( \Text{rank}[\bold{G}_{21}^n\bold{V}_{11}^n \quad \bold{G}_{22}^n\bold{V}_{12}^n \quad \bold{G}_{21}^n\bold{V}_{21}^n \quad \bold{G}_{22}^n\bold{V}_{22}^n] -\Text{rank}[\bold{G}_{21}^n\bold{V}_{11}^n \quad \bold{G}_{22}^n\bold{V}_{12}^n])         \nonumber\\
\stackrel{(a)}{\leq}&  \Text{rank}[\bold{G}_{11}^n\bold{V}_{11}^n \quad \bold{G}_{12}^n\bold{V}_{12}^n]  +\frac{3}{2}  \Text{rank}[\bold{G}_{21}^n\bold{V}_{11}^n \quad \bold{G}_{22}^n\bold{V}_{12}^n \quad \bold{G}_{21}^n\bold{V}_{21}^n \quad \bold{G}_{22}^n\bold{V}_{22}^n]   \nonumber\\
& -\frac{3}{2}   \Text{rank}[\bold{G}_{21}^n\bold{V}_{11}^n \quad \bold{G}_{22}^n\bold{V}_{12}^n]\nonumber \\
\substack{\text{(Lemma~\ref{lemma1})}   \\ a.s. \\ \leq}&  \frac{3}{2}  \Text{rank}[\bold{G}_{21}^n\bold{V}_{11}^n \quad \bold{G}_{22}^n\bold{V}_{12}^n \quad \bold{G}_{21}^n\bold{V}_{21}^n \quad \bold{G}_{22}^n\bold{V}_{22}^n]  \nonumber\\
 \leq & \frac{3}{2} n, \label{converse}
\end{align}
where (a) follows from the fact that
$\Text{rank}[\bold{A}\quad \bold{B}] \leq \Text{rank}[\bold{A}] + \Text{rank} [\bold{B}]$.
Therefore, by dividing both sides of the inequality in (\ref{converse}) by $n$, and letting $n\to\infty$ we get
\begin{equation}
d_{11}+d_{12}+\frac{3}{2}(d_{21}+d_{22})\leq \frac{3}{2}.
\end{equation}
Hence, the proof of converse for Theorem \ref{theorem1} is complete.
$  \blacksquare$

We will next prove Lemma~\ref{lemma1}.

\subsection{Proof of Lemma~1}\label{sec:keyLemmaProof}

Let us fix $n\in \mathbb N$, and 
consider a fixed linear coding strategy $\{  f_1^{(n)},f_2^{(n)} \}$, with corresponding $ \bold{V}_{11}^{n},\bold{V}_{12}^{n} $ as defined in (\ref{coding}). For notational simplicity in the proof, we denote $\bold{V}_{11}^n$ by $\bold{V}_{1}^n$, and $\bold{V}_{12}^n$ by $\bold{V}_{2}^n$. We first state some definitions.

\begin{definition}  \label{T} Consider a fixed linear coding strategy $\{  f_1^{(n)},f_2^{(n)} \}$, with corresponding $ \bold{V}_{1}^{n}\stackrel{\Delta}{=} \bold{V}_{11}^{n},\bold{V}_{2}^{n} \stackrel{\Delta}{=} \bold{V}_{12}^{n}$. 
Define the random set $ \bm{\mathcal T}_{\{  f_1^{(n)},f_2^{(n)} \}} (\bm{\mathcal  G^n})$ with its alphabet being the power set of $\{ 1,2,\ldots, n\}$ as follows.
For any realization of channels $\bm{\mathcal  G^n}=\mathcal G^n$, which results in  $\bold{G}_{21}^n=G_{21}^n, \bold{G}_{22}^n=G_{22}^n,\bold{G}_{11}^n=G_{11}^n, \bold{G}_{12}^n=G_{12}^n$,   and $\bold{V}_1^n=V_1^{n}, \bold{V}_2^n=V_2^{n}$, we define
\begin{equation}
 \mathcal T_{\{  f_1^{(n)},f_2^{(n)} \}}(\mathcal G^n) \triangleq  \{ t| [\vec { v}_1(t)^\top \quad \vec 0_{1\times m_2(n)}],[\vec 0_{1\times m_1(n)}\quad \vec { v}_2(t)^\top]  \in   \text{rowspan} [G_{21}^{t-1}V_1^{t-1}  \quad  G_{22}^{t-1}V_2^{t-1}]  \} .
\end{equation}
\end{definition}
In  words, $ \bm{\mathcal T}_{\{  f_1^{(n)},f_2^{(n)} \}} (\bm{\mathcal  G^n})$ represents the set of random timeslots (random due to the randomness in channels), where the beamforming vectors transmitted by the two transmitters are already individually  recoverable by  $\Text{Rx}_2$ using its received beamforming vectors in  the previous timeslots.
Since the code $\{  f_1^{(n)},f_2^{(n)} \}$ is fixed in the proof, for notational simplicity from now on we denote $ \bm{\mathcal T}_{\{  f_1^{(n)},f_2^{(n)} \}} (\bm{\mathcal  G^n})$ by $\bm{\mathcal T}$.

\begin{definition}   \label{r1} Consider a fixed linear coding strategy $\{  f_1^{(n)},f_2^{(n)} \}$, with corresponding $ \bold{V}_{1}^{n}\stackrel{\Delta}{=} \bold{V}_{11}^{n},\bold{V}_{2}^{n} \stackrel{\Delta}{=} \bold{V}_{12}^{n}$. 
Define random variables $\bold{r}_1(\bm{\mathcal{G}}^n), \bold{r}_2(\bm{\mathcal {G}}^n)  $  in $\{1,\ldots,n\}$ as follows. For any realization of channels $\bm{\mathcal  G^n}=\mathcal G^n$, which results in  $\bold{G}_{21}^n=G_{21}^n, \bold{G}_{22}^n=G_{22}^n,\bold{G}_{11}^n=G_{11}^n, \bold{G}_{12}^n=G_{12}^n$,   and $\bold{V}_1^n=V_1^{n}, \bold{V}_2^n=V_2^{n}$, define
\begin{align}
& r_i(\mathcal G^n)  \triangleq  \Text{dim} \left( \Text{span}(  \mathcal E_i(\mathcal G^n)  ) \right), \quad i=1,2,\nonumber
\end{align}
where 
\begin{align}
& \mathcal E_1(\mathcal G^n) \triangleq  \{ \vec s_{m_1(n)\times 1} |\quad \exists \vec l_{ n \times 1} \quad s.t.\quad  [\vec s^\top\quad \vec 0_{1\times m_2(n)}]=\vec l~^\top [G_{21}^nV_1^n \quad G_{22}^n V_2^n]  \}    \nonumber\\
& \mathcal E_2(\mathcal G^n) \triangleq  \{  \vec s_{ m_2(n) \times 1} |\quad \exists \vec l_{ n \times 1} \quad s.t.\quad  [\vec 0_{1\times m_1(n)} \quad \vec s^\top]=\vec l~^\top [G_{21}^nV_1^n \quad G_{22}^n V_2^n]  \}    \nonumber .
\end{align}
\end{definition}
In  words, $ \bold{r}_1(\bm{\mathcal G}^n)$ can be interpreted as the number of linearly independent equations that $\Text{Rx}_2$ can recover from its received signal, which only involve symbols of $\Text{Tx}_1$.
Hereafter, we denote $\bold{r}_1(\bm{\mathcal{G}}^n), \bold{r}_2(\bm{\mathcal {G}}^n)  $ simply by $\bold{r}_1, \bold{r}_2 $.

We will now state the following lemma, proved in Appendix~\ref{lemm2}, which is the key to proving Lemma~\ref{lemma1}.

\begin{lemma}\label{main1}
For any linear coding strategy $\{  f_1^{(n)},f_2^{(n)} \}$, with corresponding $ \bold{V}_{1}^{n}\stackrel{\Delta}{=} \bold{V}_{11}^{n},\bold{V}_{2}^{n} \stackrel{\Delta}{=} \bold{V}_{12}^{n}$  defined in (\ref{coding}),
\begin{itemize}
\item
$ \Text{rank}[\bold{G}_{11}^n\bold{V}_1^n\quad \bold{G}_{12}^n\bold{V}_2^n] - \Text{rank}[\bold{G}_{21}^n\bold{V}_1^n \quad \bold{G}_{22}^n\bold{V}_2^n]\stackrel{a.s.}{ \leq}  \Text{rank}[\bold{G}_{11}^{\bm{\mathcal T}}\bold{V}_1^{\bm{\mathcal T}}\quad \bold{G}_{12}^{\bm{\mathcal T}}\bold{V}_2^{\bm{\mathcal T}}] $

\item
$\Text{rank}[\bold{V}_j^{\bm{\mathcal T}}] \stackrel{}{\leq} \bold{r}_j, \quad j=1,2$

\item
$ \bold{r}_j \stackrel{a.s.}{\leq}   \Text{rank}[\bold{G}_{21}^n\bold{V}_1^n \quad \bold{G}_{22}^n\bold{V}_2^n]-\Text{rank}[\bold{V}_{3-j}^n] , \quad  j=1,2$
\end{itemize}
where $\bm{\mathcal T}$ is defined in Definition \ref{T}, $\bold{V}_i^{\bm{\mathcal T}}$ represents the random sub-matrix of $\bold{V}_i^n$ derived by keeping rows whose indices are in $\bm{\mathcal T}$, and $\bold{r}_1,\bold{r}_2$ are defined in Definition \ref{r1}.
\end{lemma}

\begin{remark}
Note that the first inequality in the above lemma intuitively implies that, in order to bound the difference of the dimensions of received linear subspaces at the two receivers, we only needs to focus on the timeslots in which $\Text{Rx}_2$ already knows \underline{both} of the individual transmit equations.
\end{remark}

%

We are now ready to prove Lemma~\ref{lemma1}.
We will first use Lemma \ref{main1} to find an upper bound on the difference between
$\Text{rank}[\bold{G}_{11}^n\bold{V}_1^n\quad \bold{G}_{12}^n\bold{V}_2^n]$
and $\Text{rank}[\bold{G}_{21}^n\bold{V}_1^n \quad \bold{G}_{22}^n\bold{V}_2^n]$.
\begin{align*}
& \Text{rank}[\bold{G}_{11}^n\bold{V}_1^n\quad \bold{G}_{12}^n\bold{V}_2^n] - \Text{rank}[\bold{G}_{21}^n\bold{V}_1^n \quad \bold{G}_{22}^n\bold{V}_2^n]\substack{ (\text{Lemma \ref{main1}})   \\ a.s.\\ \leq}  \Text{rank}[\bold{G}_{11}^{\bm{\mathcal T}}\bold{V}_1^{\bm{\mathcal T}}\quad \bold{G}_{12}^{\bm{\mathcal T}}\bold{V}_2^{\bm{\mathcal T}}]\\
& \leq  \Text{rank}[\bold{G}_{11}^{\bm{\mathcal T}}\bold{V}_1^{\bm{\mathcal T}}]+\Text{rank}[\bold{G}_{12}^{\bm{\mathcal T}} \bold{V}_2^{\bm{\mathcal T}}] \stackrel{a.s.}{=}\Text{rank}[\bold{V}_1^{\bm{\mathcal T}}]+\Text{rank}[\bold{V}_2^{\bm{\mathcal T}}] \\
& \stackrel{(\text{Lemma } \ref{main1})}{\leq }  \bold{r}_1 + \bold{r}_2 \\
& \substack{ (\text{Lemma \ref{main1}})   \\ a.s.\\ \leq} \Text{rank}[\bold{G}_{21}^n\bold{V}_1^n \quad \bold{G}_{22}^n\bold{V}_2^n]-\Text{rank}[\bold{V}_2^n]+  \Text{rank}[\bold{G}_{21}^n\bold{V}_1^n \quad \bold{G}_{22}^n\bold{V}_2^n]-\Text{rank}[\bold{V}_1^n] \\
& \substack{  a.s.\\ =} 2\Text{rank}[\bold{G}_{21}^n\bold{V}_1^n \quad \bold{G}_{22}^n\bold{V}_2^n]-\Text{rank}[\bold{G}_{11}^n\bold{V}_1^n]  -  \Text{rank}[\bold{G}_{12}^n\bold{V}_2^n]  \\
& \substack{  a.s.\\ \leq} 2 \Text{rank}[\bold{G}_{21}^n\bold{V}_1^n \quad \bold{G}_{22}^n\bold{V}_2^n]-\Text{rank}[\bold{G}_{11}^n\bold{V}_1^n\quad \bold{G}_{12}^n\bold{V}_2^n].
\end{align*}
By rearranging the two sides of the above inequality,  the proof of Lemma~\ref{lemma1} would be complete.


\section{The Three-User Interference Channel with Delayed CSIT}  
~\label{threeuser}

In this section we give an example that shows how Lemma~\ref{lemma1} can be useful for deriving outer bounds in other scenarios. In particular, we
utilize Lemma~\ref{lemma1} to provide a new outer bound on the three-user interference channel with delayed CSIT depicted in Fig.~\ref{3IC}. 
\begin{figure}[h!]
\centering
\includegraphics[scale=.3, trim= 10mm 10mm 10mm 10mm]{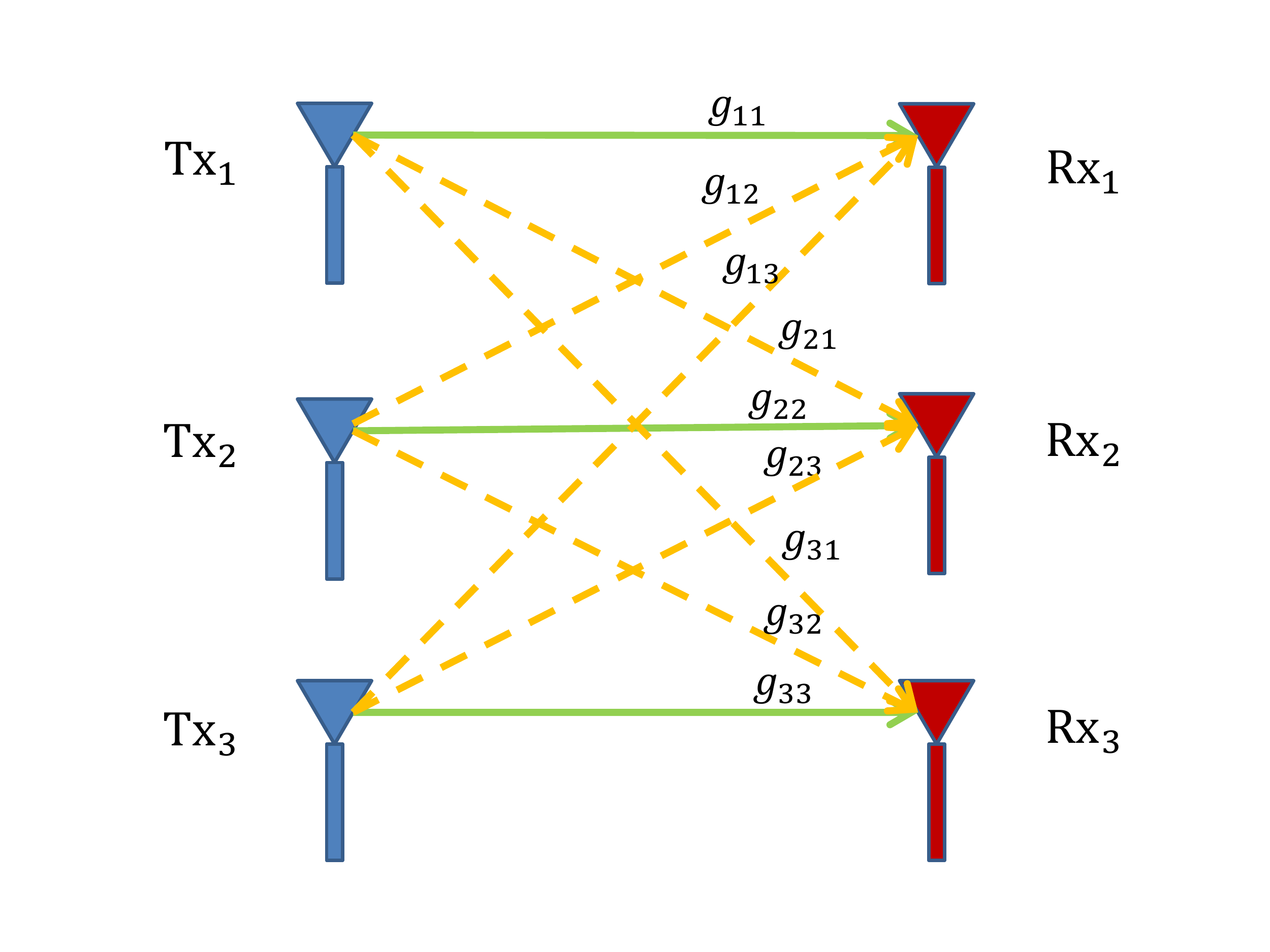}\\
\caption{Network configuration for the three-user IC.
There are three transmitters and three receivers, and for $j=1,2,3$, $\text{Tx}_j$ has  message for $\text{Rx}_j$. We assume time-varying channels, with delayed CSIT.
}\label{3IC}
\end{figure}
The channel model is similar to that of the X-channel except the channel input-output relation and decodability constraints. The received signal at $\text{Rx}_{k}$ ($k\in \{1,2,3\}$) at time $t$ is given by 
\begin{equation}
\bold{y}_{k}(t)=\sum_{j=1}^{3} \bold{g}_{kj}(t)\bold{x}_{j}(t)+\bold{z}_{k}(t).
\end{equation}

For block length of $n$ and $j=1,2,3$, we consider the decodability constraint of
\begin{equation}
\Text{dim} \left( \text{Proj}_{\bm{\mathcal I}_j^c} \text{colspan} (\bold{G}_{jj}^n \bold{V}_{j}^n) \right) =\Text{dim} \left( \text{colspan} (\bold{V}_{j}^n) \right) =m_j(n),
\label{decode3user}
\end{equation}
where $\bm{\mathcal I}_j= \cup_{i\ne j} \text{colspan} (\bold{G}_{ji}^n \bold{V}_{i}^n)$.
Denote the linear degrees of freedom  region $\mathcal D_{\text{3UserIC}}$ as the closure of the set of all achievable 3-tuples $(d_1,d_2,d_3)$, where $d_{j}=\lim_{n\to\infty} \frac{m_j(n)}{n}$, and $\{  m_1{(n)},m_2{(n)},m_3{(n)} \}$ are linearly achievable with probability $1$ for every $n\in\mathbb N$.
We are interested in characterizing the sum linear degrees of freedom:
\begin{equation}
\Text{DoF}_{\textrm{L-sum}}=\max\sum_{j=1}^{3}d_j,\qquad s.t.\quad (d_1,d_2,d_3)\in\mathcal D.
\end{equation}

With delayed CSIT, it was shown in \cite{retrospective} that the sum DoF of $\frac{9}{8}$ can be achieved, which was later improved to $\frac{36}{31}$ in \cite{abdoli}.
However, the best known outer bound so far is $\frac{3}{2}$, which also holds for the case of instantaneous CSIT~\cite{alignment}. The following theorem provides a tighter bound on the linear degrees of freedom.

\begin{theorem}
\label{theorem2}
For the three-user interference channel with delayed CSIT,
\begin{equation}
\Text{DoF}_{\Text{L-sum}} \leq \frac{9}{7}.
\end{equation}
\end{theorem}

\begin{proof}
Let us denote the symmetric degrees of freedom for three-user interference channel by $\Text{DoF}_{\textrm{L-sym}}$.
Note that due to symmetry of topology,
\begin{equation}
\Text{DoF}_{\textrm{L-sum}}=    3\times \Text{DoF}_{\textrm{L-sym}}.
\end{equation}
Hence, in order to prove the theorem it suffices to show that $\Text{DoF}_{\textrm{L-sym}} \leq \frac{3}{7}$.
So assume that for a given block length $n$, $m_1(n)=m_2(n)=m_3(n)$, and we seek to show that
if decodability is accomplished with probability 1, we should have $m_1(n)\stackrel{}{\leq} \frac{3}{7}n$.
By Lemma \ref{eq} if the decodability constraints in (\ref{decode3user}) are satisfied with probability 1 for pairs $\text{Tx}_1$-$\text{Rx}_1$ and $\text{Tx}_2$-$\text{Rx}_2$, then
\begin{align}
& \Text{rank}[ \bold{G}_{12}^n\bold{V}_2^n \quad   \bold{G}_{13}^n\bold{V}_3^n]+\Text{rank}[\bold{G}_{11}^n\bold{V}_1^n]\stackrel{a.s.}{=}\Text{rank}[\bold{G}_{11}^n\bold{V}_1^n \quad \bold{G}_{12}^n\bold{V}_2^n \quad \bold{G}_{13}^n\bold{V}_3^n ],   \label{do3user}\\
& \Text{rank}[ \bold{G}_{21}^n\bold{V}_1^n \quad   \bold{G}_{23}^n\bold{V}_3^n]+\Text{rank}[\bold{G}_{22}^n\bold{V}_2^n]\stackrel{a.s.}{=}\Text{rank}[\bold{G}_{21}^n\bold{V}_1^n \quad \bold{G}_{22}^n\bold{V}_2^n \quad \bold{G}_{23}^n\bold{V}_3^n ], \label{do3user0}
\end{align}
where $\Text{rank}[\bold{V}_1^n] \stackrel{a.s.}{=} \Text{rank}[\bold{V}_2^n]  \stackrel{a.s.}{=}  \Text{rank}[\bold{V}_3^n] \stackrel{a.s.}{=}  m_1(n)$.
Thus, assuming $m_1(n)=m_2(n)=m_3(n)$ are linearly achievable with probability 1,
 from (\ref{do3user0}), we have
\begin{align}
 \Text{rank}[\bold{G}_{22}^n\bold{V}_2^n] & \stackrel{a.s.}{=} \Text{rank}[\bold{G}_{21}^n\bold{V}_1^n \quad \bold{G}_{22}^n\bold{V}_2^n \quad \bold{G}_{23}^n\bold{V}_3^n ] - \Text{rank}[ \bold{G}_{21}^n\bold{V}_1^n \quad   \bold{G}_{23}^n\bold{V}_3^n] \nonumber \\
&  \stackrel{(a)}{ \leq}   \Text{rank}[\bold{G}_{22}^n\bold{V}_2^n \quad \bold{G}_{23}^n\bold{V}_3^n ] - \Text{rank}[    \bold{G}_{23}^n\bold{V}_3^n],  \label{do3user2}
\end{align}
where (a) follows from sub-modularity of rank (Lemma \ref{ranksubmod}).
In addition, we know that
\begin{align}
& \Text{rank}[\bold{G}_{22}^n\bold{V}_2^n] \geq   \Text{rank}[\bold{G}_{22}^n\bold{V}_2^n \quad \bold{G}_{23}^n\bold{V}_3^n ] - \Text{rank}[    \bold{G}_{23}^n\bold{V}_3^n].  \label{do3user22}
\end{align}
By (\ref{do3user2}), (\ref{do3user22}) we conclude that
\begin{align}
& \Text{rank}[\bold{G}_{22}^n\bold{V}_2^n \quad \bold{G}_{23}^n\bold{V}_3^n] \stackrel{a.s.}{=} \Text{rank}[ \bold{G}_{22}^n\bold{V}_2^n ]+\Text{rank}[\bold{G}_{23}^n\bold{V}_3^n]\stackrel{a.s.}{=} \Text{rank}[ \bold{V}_2^n ]+\Text{rank}[\bold{V}_3^n] \stackrel{a.s.}{=} 2m_1(n) .   \label{do3user3}
\end{align}
On the other hand, from Lemma~\ref{lemma1} we know that
\begin{equation}
\Text{rank}[\bold{G}_{22}^n\bold{V}_2^n \quad \bold{G}_{23}^n\bold{V}_3^n] \stackrel{a.s.}{\leq} \frac{3}{2}   \Text{rank}[ \bold{G}_{12}^n\bold{V}_2^n \quad   \bold{G}_{13}^n\bold{V}_3^n] .\label{do3user4}
\end{equation}
Hence, by (\ref{do3user3}), (\ref{do3user4}),
\begin{equation}
  \Text{rank}[ \bold{G}_{12}^n\bold{V}_2^n \quad   \bold{G}_{13}^n\bold{V}_3^n]  \stackrel{a.s.}{\geq} \frac{4}{3} m_1(n).\label{do3user5}
\end{equation}

Finally,  by considering (\ref{do3user}), (\ref{do3user5}), and the fact that
$\Text{rank}[\bold{G}_{11}^n\bold{V}_1^n \quad \bold{G}_{12}^n\bold{V}_2^n \quad \bold{G}_{13}^n\bold{V}_3^n ] \leq n$,
we get
\begin{eqnarray*}
m_1(n) \stackrel{a.s.}{=} \Text{rank}[\bold{V}_1]  \stackrel{a.s.}{=} \Text{rank}[\bold{G}_{11}^n\bold{V}_1^n] &    \substack{ \text{(\ref{do3user})} \\ a.s  \\ =} & \Text{rank}[\bold{G}_{11}^n\bold{V}_1^n \quad \bold{G}_{12}^n\bold{V}_2^n \quad \bold{G}_{13}^n\bold{V}_3^n ] - \Text{rank}[ \bold{G}_{12}^n\bold{V}_2^n \quad   \bold{G}_{13}^n\bold{V}_3^n] \\
&   \substack{ \text{(\ref{do3user5})} \\  a.s.  \\ \leq}  &  n - \frac{4}{3} m_1(n),
\end{eqnarray*}
which implies that $m_1(n) \stackrel{}{\leq} \frac{3}{7} n$ because $n, m_1(n)$ are non-random, and this completes the proof.
\end{proof}

\section{Concluding Remarks and Future Directions} 
\label{discussion}

In this paper, we characterized the linear degrees of freedom of the $X$-channel with delayed CSIT. Our main contribution was the development of a general lemma that shows that, if two distributed transmitters employ linear strategies, the ratio of  the dimensions of received linear subspaces at the two receivers cannot exceed $\frac{3}{2}$, due to lack of instantaneous  knowledge of the channels. We also applied this general lemma to the three-user interference channel with delayed CSIT, thereby deriving a new upper bound of $\frac{9}{7}$ on its linear degrees of freedom. 

We conjecture that the total degrees of freedom of the $X$-channel with delayed CSIT (without restriction to linear schemes) is also $\frac{6}{5}$. In fact, we conjecture the following generalization of Lemma~\ref{lemma1} for general encoding strategies.

\begin{conjecture}
\label{conj}
Consider the  2-transmitter 2-receiver network  setting of Lemma~\ref{lemma1}. For any $n\in \mathbb N$ and any coding strategy denoted by encoding functions $\{  f_1^{(n)},f_2^{(n)} \}$, and its corresponding received signals,  $\vec {\bf y}_1^n$ and $\vec {\bf y}_2^n$, we have
\begin{equation}
\label{ineqConj}
h(\vec {\bf y}_1^n| \bm{\mathcal {G}}^n) \leq \frac{3}{2} h(\vec {\bf y}_2^n | \bm{\mathcal {G}}^n) +n\times o(\log(P)).\end{equation}
\end{conjecture}

Therefore, a future direction would be to remove the linearity restriction on the encoding schemes, and prove (or disprove) the above conjecture, which (if true) will lead to the  DoF characterization of the $X$-channel with delayed CSIT.

We also believe that similar techniques could be applied to other important network configurations to gain insight on how delayed CSIT can be used to improve the Degrees of Freedom, and what the limitations on this DoF improvement are. In particular the $K$-user interference channel and multi-hop interference networks (e.g.,~\cite{SAKKK,AAMultiHopC,AAMultiHopJ}), in which there is a large gap between the state-of-the-art inner and outer bounds on DoF with delayed CSIT, can be considered.

Most research so far has focused on understanding the impact of delayed CSIT via a coarse DoF analysis. In the context of linear schemes, this  can be viewed as understanding the impact of delayed CSIT on the \emph{dimension} of desired signal spaces at the receivers of a wireless networks. While such analysis provides a first-order understanding of the impact of delayed CSIT on capacity, it is of great value to refine the analysis and study the impact of delayed CSIT on the \emph{volume} of desired signal spaces at the receivers. A first step along this direction has been taken in~\cite{MISOBCAllerton} to approximate the capacity of MISO BC with delayed CSIT to within 1 bit/sec/Hz.

\bibliographystyle{IEEEbib}
\bibliography{bib_delayedCSIT}

\begin{thebibliography}{10}

\bibitem{Ours}
S.~Lashgari, A.~S. Avestimehr, and C.~Suh,
\newblock ``{A rank ratio inequality and the linear degrees of freedom of
  X-channel with delayed CSIT},''
\newblock {\em Allerton Conference on Communicaiton, Control, and Computing},
  2013.

\bibitem{mimox}
S.~A. Jafar and S.~Shamai,
\newblock ``{Degrees of freedom region of the MIMO $X$ channel},''
\newblock {\em IEEE Transactions on Information Theory}, vol. 54, no. 1, pp.
  151--170, Jan. 2008.

\bibitem{IA}
M.~A. Maddah-Ali, A.~S. Motahari, and A.~K. Khandani,
\newblock ``{Communication over {MIMO} {$X$} channels: interference alignment,
  decomposition, and performance analysis},''
\newblock {\em IEEE Transactions on Information Theory}, vol. 54, no. 8, pp.
  3457--3470, Aug. 2008.

\bibitem{alignment}
Viveck~R. Cadambe and Syed~A. Jafar,
\newblock ``{Interference alignment and the degree of freedom for the $K$ user
  interference channel},''
\newblock {\em IEEE Transactions on Information Theory}, vol. 54, no. 8, pp.
  3425--3441, Aug. 2008.

\bibitem{MAT}
M.~A. Maddah-Ali and D.~N. Tse,
\newblock ``{Completely stale transmitter channel state information is still
  very useful},''
\newblock {\em IEEE Transactions on Information Theory}, vol. 58, July 2012.

\bibitem{retrospective}
H.~Maleki, S.~A. Jafar, and S.~Shamai,
\newblock ``{Retrospective interference alignment over interference
  networks},''
\newblock {\em IEEE Journal of Selected Topics in Signal Processing}, vol. 6,
  no. 3, pp. 228--240, June 2012.

\bibitem{varanasi}
C.~S. Vaze and M.~K. Varanasi,
\newblock ``{The degrees of freedom region and intererence alingment for the
  MIMO inteference channel with delayed CSIT},''
\newblock {\em IEEE Transactions on Information Theory}, vol. 58, no. 7, pp.
  4396--4417, July 2012.

\bibitem{abdoli}
M.J. Abdoli, A.~Ghasemi, and A.K. Khandani,
\newblock ``{On the degrees of freedom of $K$-user SISO interference and $X$
  channels with delayed CSIT},''
\newblock {\em arXiv:1109.4314}, 2011.

\bibitem{xchannel}
A.~Ghasemi, A.S. Motahari, and A.K. Khandani,
\newblock ``{On the degrees of freedom of $X$ channel with delayed CSIT},''
\newblock {\em IEEE International Symposium on Information Theory}, 2011.

\bibitem{VMABinaryIC}
A.~Vahid, M.~Maddah-Ali, and A.~S. Avestimehr,
\newblock ``{Capacity results for binary fading interference channels with
  delayed CSIT},''
\newblock {\em arXiv:1301.5309}, 2013.

\bibitem{vahid2013communication}
A.~Vahid, M.A. Maddah-Ali, and A.S. Avestimehr,
\newblock ``Communication through collisions: opportunistic utilization of past
  receptions,''
\newblock {\em accepted for publication in IEEE Infocom 2014. arXiv preprint
  arXiv:1312.0116}, 2013.

\bibitem{ElGamal}
Abbas~El Gamal,
\newblock ``{The feedback capacity of degraded broadcast channels},''
\newblock {\em IEEE Transactions on Information Theory}, vol. 24, no. 3, pp.
  379--381, May 1978.

\bibitem{MISOBCAllerton}
A.~Vahid, M.A. Maddah-Ali, and A.S. Avestimehr,
\newblock ``Approximate capacity of the two-user {MISO} broadcast channel with
  delayed {CSIT},''
\newblock {\em Fifty-First Annual Allerton Conference on Communication,
  Control, and Computing}, 2013.

\bibitem{Tassiulas}
L.~Georgiadis and L.~Tassiulas,
\newblock ``{Broadcast erasure channel with feedback - capacity and
  algorithms},''
\newblock {\em Worshop on Network Coding, Theory, and Applications}, pp.
  54--61, 2009.

\bibitem{Yang}
S.~Yang, M.~Kobayashi, D.~Gesbert, and X.~Yi,
\newblock ``{Degrees of freedom of time correlated MISO broadcast channel with
  delayed CSIT},''
\newblock {\em IEEE Transactions on Information Theory}, vol. 59, 2012.

\bibitem{Extremal}
T.~Liu and P.~Viswanath,
\newblock ``{An extremal inequality motivated by multiterminal
  information-theoretic problems},''
\newblock {\em IEEE Transactions on Information Theory}, vol. 53, pp.
  1839--1851, 2007.

\bibitem{bresler}
G.~Bresler, D.~Cartwright, and D.~N. Tse,
\newblock ``{Interference alignment for the MIMO interference channel},''
\newblock {\em arXiv:1303.5678}, 2013.

\bibitem{xjafar}
V.~R. Cadambe and S.~A. Jafar,
\newblock ``{Degrees of freedom of wireless $X$ networks},''
\newblock {\em IEEE International Symposium on Information Theory}, 2008.

\bibitem{Lovasz}
L.~Lovasz,
\newblock ``{Submodular functions and convexity},''
\newblock {\em Springer}, 1983.

\bibitem{SAKKK}
I.~{Shomorony} and A.~S. {Avestimehr},
\newblock ``Degrees of freedom of two-hop wireless networks: everyone gets the
  entire cake,''
\newblock {\em accepted for publication in IEEE Transactions on Information
  Theory}, 2013.

\bibitem{AAMultiHopC}
J.~Abdoli and A.~S. Avestimehr,
\newblock ``{On degrees of freedom scaling in layered interference networks
  with delayed CSI},''
\newblock {\em IEEE International Symposium on Information Theory}, 2013.

\bibitem{AAMultiHopJ}
J.~Abdoli and A.~S. Avestimehr,
\newblock ``{Layered interference networks with delayed CSI: DoF scaling with
  distributed transmitters},''
\newblock {\em accepted for publication in IEEE Transactions on Information
  Theory}, 2013.

\end{thebibliography}

\begin{appendices}

\section{Proof of Lemma \ref{main1}} \label{lemm2}

\subsection{Proof of $ \Text{rank}[\bold{G}_{11}^n\bold{V}_1^n\quad \bold{G}_{12}^n\bold{V}_2^n] - \Text{rank}[\bold{G}_{21}^n\bold{V}_1^n \quad \bold{G}_{22}^n\bold{V}_2^n]\stackrel{a.s.}{ \leq}  \Text{rank}[\bold{G}_{11}^{\bm{\mathcal T}}\bold{V}_1^{\bm{\mathcal T}}\quad \bold{G}_{12}^{\bm{\mathcal T}}\bold{V}_2^{\bm{\mathcal T}}]$:}

For a fixed
linear coding strategy $\{  f_1^{(n)},f_2^{(n)} \}$, with corresponding $ \bold{V}_1^{n},\bold{V}_2^{n} $,
let $\mathcal  A_i,\mathcal B_i,\mathcal C_i$, $i=1,2,\ldots, n$, denote the following sets:
\begin{itemize}
\item $\mathcal A_i\triangleq \{ \mathcal G^n| \quad \Text{rank}[G_{21}^{i}V_1^{i} \quad G_{22}^{i}V_2^{i}]=\Text{rank}[G_{21}^{i-1}V_1^{i-1} \quad G_{22}^{i-1}V_2^{i-1}]  \} . $
\item $\mathcal B_i\triangleq \{ \mathcal G^n|  \quad [\vec v_1(i)^\top \quad \vec 0_{1\times m_2(n)}],[\vec 0_{1\times m_1(n)}\quad \vec v_2(i)^\top]  \in   \text{rowspan} [G_{21}^{i-1}V_1^{i-1}  \quad  G_{22}^{i-1}V_2^{i-1}] \}  $.
\item $\mathcal C_i\triangleq \{ \mathcal G^n|  \quad  \Text{rank}[G_{11}^i V_1^i\quad G_{12}^iV_2^i] =\Text{rank}[G_{11}^{i-1}V_1^{i-1}\quad G_{12}^{i-1}V_2^{i-1}]+1 \}$,
\end{itemize}
Note that $\mathcal B_i$ is equivalent to  $\{ \mathcal G^n|  \quad  i\in \mathcal T( \mathcal G^n) \}$.
In order to prove Lemma \ref{main1} we first state the following lemma, whose proof is postponed to Appendix \ref{lemm6}.

\begin{lemma}\label{main5}
\begin{equation}
\Pr (\bm{\mathcal G}^n \in \cup_{i=1}^{n}\mathcal (\mathcal A_i \cap \mathcal B_i^c))=0.
\end{equation}
\end{lemma}

Lemma~\ref{main5} implies that we need to prove the first inequality in Lemma \ref{main1} only for channel realizations $\bm{\mathcal G}^n=\mathcal G^n$, such that $\mathcal G^n \notin \cup_{i=1}^{n}\mathcal (\mathcal A_i \cap \mathcal B_i^c)$ (since, the rest have probability measure zero). Thus,  we only need to show that for any arbitrary channel realization $\bm{\mathcal G}^n=\mathcal G^n$ with the corresponding beamforming matrices $V_1^n,V_2^n$, and $\bm{\mathcal T}=\mathcal T$, such that $\mathcal G^n \notin \cup_{i=1}^{n}\mathcal (\mathcal A_i \cap \mathcal B_i^c)$,  we have \begin{equation} \label{eq:mainIneq} \Text{rank}[G_{11}^nV_1^n\quad G_{12}^nV_2^n] - \Text{rank}[G_{21}^nV_1^n \quad G_{22}^nV_2^n]\stackrel{}{ \leq}  \Text{rank}[G_{11}^{\mathcal T} V_1^{\mathcal T}\quad G_{12}^{\mathcal T} V_2^{\mathcal T}]. \end{equation}
Let $I(.)$ denote the indicator function, we now bound the left hand side of (\ref{eq:mainIneq}) as follows.
\begin{align}
&\Text{rank}[G_{11}^nV_1^n\quad G_{12}^nV_2^n] - \Text{rank}[G_{21}^nV_1^n \quad G_{22}^nV_2^n] \nonumber \\
&= \sum_{i=1}^{n} ( \Text{rank}[G_{11}^iV_1^i\quad G_{12}^iV_2^i] -\Text{rank}[G_{11}^{i-1}V_1^{i-1}\quad G_{12}^{i-1}V_2^{i-1}])- (\Text{rank}[G_{21}^iV_1^i \quad G_{22}^iV_2^i]-\Text{rank}[G_{21}^{i-1}V_1^{i-1} \quad G_{22}^{i-1}V_2^{i-1}]) \nonumber \\
&\leq \sum_{i=1}^{n}\max\{ ( \Text{rank}[G_{11}^iV_1^i\quad G_{12}^iV_2^i] -\Text{rank}[G_{11}^{i-1}V_1^{i-1}\quad G_{12}^{i-1}V_2^{i-1}])\nonumber\\
& - (\Text{rank}[G_{21}^iV_1^i \quad G_{22}^iV_2^i]-\Text{rank}[G_{21}^{i-1}V_1^{i-1} \quad G_{22}^{i-1}V_2^{i-1}]),0\}\nonumber \\
&\stackrel{(a)}{=}\sum_{i=1}^{n} I( \Text{rank}[G_{11}^iV_1^i\quad G_{12}^iV_2^i] =\Text{rank}[G_{11}^{i-1}V_1^{i-1}\quad G_{12}^{i-1}V_2^{i-1}]+1)\nonumber\\
& \times I(\Text{rank}[G_{21}^iV_1^i \quad G_{22}^iV_2^i]=\Text{rank}[G_{21}^{i-1}V_1^{i-1} \quad G_{22}^{i-1}V_2^{i-1}])   \nonumber \\
& =  \sum_{i=1}^{n} I(\mathcal G^n \in \mathcal A_i\cap \mathcal C_i)  =\sum_{i=1}^{n} (I(\mathcal G^n \in\mathcal A_i\cap\mathcal B_i \cap \mathcal C_i) + I(\mathcal G^n \in\mathcal A_i\cap \mathcal B_i^c\cap\mathcal C_i)) \nonumber \\
& \leq \sum_{i=1}^{n} (I(\mathcal G^n \in\mathcal B_i\cap\mathcal C_i) + I(\mathcal G^n \in\mathcal A_i\cap\mathcal B_i^c) )
\stackrel{(b)}{=} \sum_{i=1}^{n}I(\mathcal G^n \in \mathcal B_i\cap\mathcal C_i) \stackrel{(c)}{=} \sum_{i\in \mathcal T}^{} I( \mathcal G^n \in \mathcal C_i)\nonumber \\
&= \sum_{i\in \mathcal T}^{} I(\Text{rank}[G_{11}^iV_1^i\quad G_{12}^iV_2^i] =\Text{rank}[G_{11}^{i-1}V_1^{i-1}\quad G_{12}^{i-1}V_2^{i-1}]+1)          ,\label{x1}
\end{align}
where (a) holds since $\Text{rank}[G_{k1}^iV_1^i\quad G_{k2}^iV_2^i] -\Text{rank}[G_{k1}^{i-1}V_1^{i-1}\quad G_{k2}^{i-1}V_2^{i-1}]  \in \{0,1\}$ for $k=1,2$;
and (b) follows from the assumption that $\mathcal G^n \notin (\mathcal A_i\cap \mathcal B^c_i)$ for $i\in\{1,2,\ldots, n\}$;
and (c) follows from the fact that $\mathcal T =\{i| \mathcal G^n\in \mathcal B_i\}$.
We now only need to show the following to complete the proof of (\ref{eq:mainIneq}).
\begin{align}
&\sum_{i\in \mathcal T}^{} I(\Text{rank}[G_{11}^iV_1^i\quad G_{12}^iV_2^i] =\Text{rank}[G_{11}^{i-1}V_1^{i-1}\quad G_{12}^{i-1}V_2^{i-1}]+1)    \leq \Text{rank}[G_{11}^{\mathcal T} V_1^{\mathcal T} \quad G_{12}^{\mathcal T} V_2^{\mathcal T} ]. \label{need}
\end{align}

Without loss of generality, let us assume that ${\mathcal T}=\{\tau _1,\tau _2,\ldots ,\tau _k\}$ for some $k$, such that $\tau _1<\tau _2< \ldots < \tau _k$.
We define ${\mathcal T}_j\triangleq \{\tau _1,\tau _2,\ldots, \tau _j\}$, and use $V_1^{{\mathcal T}_j}$ and $V_2^{{\mathcal T}_j}$ to denote the sub-matrices of $V_1^n$ and $V_2^n$ with rows in ${\mathcal T}_j$. We also use $G_{11}^{{\mathcal T}_j}$ to denote the $|{\mathcal T}_j|\times |{\mathcal T}_j|$ diagonal matrix with channel coefficients of $g_{11}(t)$ at timeslots $t\in {\mathcal T}_j$ on its diagonal (similarly defined for other channel matrices). We now present a claim that will be used to show (\ref{need}) and complete the proof.
\begin{claim} \label{cl3}
For any $j=1,2,\ldots, k$,
\begin{align}
& I( \Text{rank}[G_{11}^{\tau _j}V_1^{\tau _j}\quad G_{12}^{\tau _j}V_2^{\tau _j}] =\Text{rank}[G_{11}^{\tau _j-1}V_1^{\tau _j-1}\quad G_{12}^{\tau _j-1}V_2^{\tau _j-1}]+1) \nonumber \\
& \leq  I( \Text{rank}[G_{11}^{{\mathcal T}_j}V_1^{{\mathcal T}_j}\quad G_{12}^{{\mathcal T}_j}V_2^{{\mathcal T}_j}] =\Text{rank}[G_{11}^{{\mathcal T}_{j-1}}V_1^{{\mathcal T}_{j-1}}\quad G_{12}^{{\mathcal T}_{j-1}}V_2^{{\mathcal T}_{j-1}}]+1).
\end{align}
\end{claim}

\begin{proof}
The claim is trivially true when $\Text{rank}[G_{11}^{\tau _j}V_1^{\tau _j}\quad G_{12}^{\tau _j}V_2^{\tau _j}] =\Text{rank}[G_{11}^{\tau _j-1}V_1^{\tau _j-1}\quad G_{12}^{\tau _j-1}V_2^{\tau _j-1}].$
So, suppose $ \Text{rank}[G_{11}^{\tau _j}V_1^{\tau _j}\quad G_{12}^{\tau _j}V_2^{\tau _j}] =\Text{rank}[G_{11}^{\tau _j-1}V_1^{\tau _j-1}\quad G_{12}^{\tau _j-1}V_2^{\tau _j-1}]+1$.
 It means that
 $[g_{11}(\tau _j)\vec v_1(\tau _j)^\top \quad g_{12}(\tau _j)\vec v_2(\tau _j)^\top]$ is linearly independent of $\text{rowspan}[G_{11}^{\tau _j-1}V_1^{\tau _j-1}\quad G_{12}^{\tau _j-1}V_2^{\tau _j-1}]$.
 Since ${\mathcal T}_{j-1}\subseteq \{1,2,\ldots, \tau _j-1\}$, then \\
$[g_{11}(\tau _j)\vec v_1(\tau _j)^\top \quad g_{12}(\tau _j)\vec v_2(\tau _j)^\top]$ is also linearly independent of
$\text{rowspan}[G_{11}^{{\mathcal T}_{j-1}}V_1^{{\mathcal T}_{j-1}}\quad G_{12}^{{\mathcal T}_{j-1}}V_2^{{\mathcal T}_{j-1}}]$. Hence,
\begin{align}
&\Text{rank}[G_{11}^{{\mathcal T}_j}V_1^{{\mathcal T}_j}\quad G_{12}^{{\mathcal T}_j}V_2^{{\mathcal T}_j}] =\Text{rank}[G_{11}^{{\mathcal T}_{j-1}}V_1^{{\mathcal T}_{j-1}}\quad G_{12}^{{\mathcal T}_{j-1}}V_2^{{\mathcal T}_{j-1}}]+1.\nonumber
\end{align}
\end{proof}

Based on this claim, the proof of (\ref{need}) is as follows.
\begin{align}
&\sum_{i\in {\mathcal T}}^{} I(\Text{rank}[G_{11}^iV_1^i\quad G_{12}^iV_2^i] =\Text{rank}[G_{11}^{i-1}V_1^{i-1}\quad G_{12}^{i-1}V_2^{i-1}]+1) \nonumber \\
&=\sum_{j=1}^{k} I( \Text{rank}[G_{11}^{\tau _j}V_1^{\tau _j}\quad G_{12}^{\tau_j}V_2^{\tau _j}] =\Text{rank}[G_{11}^{\tau_j-1}V_1^{\tau_j-1}\quad G_{12}^{\tau_j-1}V_2^{\tau_j-1}]+1) \nonumber \\
&\stackrel{\text{Claim } \ref{cl3}}{\leq }\sum_{j=1}^{k} I( \Text{rank}[G_{11}^{{\mathcal T}_j}V_1^{{\mathcal T}_j}\quad G_{12}^{{\mathcal T}_j}V_2^{{\mathcal T}_j}] =\Text{rank}[G_{11}^{{\mathcal T}_{j-1}}V_1^{{\mathcal T}_{j-1}}\quad G_{12}^{{\mathcal T}_{j-1}}V_2^{{\mathcal T}_{j-1}}]+1)  \nonumber   \\
&=\Text{rank}[G_{11}^{{\mathcal T}_k}V_1^{{\mathcal T}_k}\quad G_{12}^{{\mathcal T}_k}V_2^{{\mathcal T}_k}]=\Text{rank}[G_{11}^{\mathcal T} V_1^{\mathcal T}\quad G_{12}^{\mathcal T} {V}_2^{\mathcal T}].  \nonumber
\end{align}


\subsection{Proof of $\Text{rank}[\bold{V}_j^{\bm{\mathcal T}}] \leq \bold{r}_j,  \quad (j=1,2):$} \label{lemm4}
It is sufficient to prove that
$\Text{rank}[\bold{V}_1^{\bm{\mathcal T}}] \leq \bold{r}_1$,
since the other inequality (i.e. $\Text{rank}[\bold{V}_2^{\bm{\mathcal T}}] \leq \bold{r}_2$) can be proven similarly.
We show that for any realization $\mathcal G^n=\{ G_{kj}^n \}_{k,j\in\{1,2\}}$ with the corresponding values $\mathcal T$, $r_1$, and matrices $V_1^n,V_2^n$, we have $\Text{rank}[V_1^{\bold{{\mathcal T}}}] \leq r_1$.
But according to definition of $r_1$, it is sufficient to prove
\begin{equation}\label{suffProof}
\text{rowspan}[V_1^{\bold{\mathcal T}} ] \subseteq \Text{span}(\vec s_{ m_1(n) \times 1}|\quad \exists \vec  l_{ n\times 1} \quad s.t.\quad  [\vec s^\top\quad  \vec 0_{1\times m_2(n)}]=\vec l~^\top [G_{21}^nV_1^n \quad G_{22}^nV_2^n]     ).
\end{equation}

The following proves (\ref{suffProof}), thereby completing the proof for $\Text{rank}[V_1^{\bold{{\mathcal T}}}] \leq r_1$:
\begin{align}
 \text{rowspan}[V_1^{\bold{\mathcal T}} ] & =  \text{span}(\vec v_1(i) | 1\leq  i \leq n,  [ \vec v_1(i)^\top \quad \vec 0_{1\times m_2(n)}], [  \vec 0_{1\times m_1(n)} \quad \vec v_2(i)^\top ]  \in \text{rowspan} [G_{21}^{i-1}V_1^{i-1}  \quad  G_{22}^{i-1}V_2^{i-1}])  \nonumber \\
&\subseteq  \text{span}(\vec v_1(i) | 1\leq  i \leq n,  [ \vec v_1(i)^\top \quad \vec 0_{1\times m_2(n)}], [  \vec 0_{1\times m_1(n)} \quad \vec v_2(i)^\top ]  \in \text{rowspan} [G_{21}^{n}V_1^{n}  \quad  G_{22}^{n}V_2^{n}]) \nonumber \\
&\subseteq  \text{span}(\vec v_1(i) | 1\leq  i \leq n,  [ \vec v_1(i)^\top \quad \vec 0_{1\times m_2(n)}]  \in \text{rowspan} [G_{21}^{n}V_1^{n}  \quad  G_{22}^{n}V_2^{n}]) \nonumber \\
&\subseteq \Text{span}(\vec s_{ m_1(n) \times 1}|\quad \exists \vec  l_{ n\times 1} \quad s.t.\quad  [\vec s^\top\quad  \vec 0_{1\times m_2(n)}]=\vec l~^\top [G_{21}^nV_1^n \quad G_{22}^nV_2^n]     ).\nonumber 
\end{align}

\subsection{Proof of $ \bold{r}_j \stackrel{a.s.}{\leq}   \Text{rank}[\bold{G}_{21}^n\bold{V}_1^n \quad \bold{G}_{22}^n\bold{V}_2^n]-\Text{rank}[\bold{V}_{3-j}^n] , \quad  (j=1,2):$} \label{lemm5}
We will show this for $j=1$, i.e., $\bold{r}_1 \stackrel{a.s.}{\leq}  \Text{rank}[\bold{G}_{21}^n\bold{V}_1^n \quad \bold{G}_{22}^n\bold{V}_2^n]-\Text{rank}[\bold{V}_2^n]$.
The proof for $j=2$ will be similar. Since $\Text{rank}[\bold{G}_{22}^n\bold{V}_2^n] \stackrel{a.s.}{=} \Text{rank}[\bold{V}_2^n]$, it is sufficient to show that $ \bold{r}_1 \stackrel{}{\leq}   \Text{rank}[\bold{G}_{21}^n\bold{V}_1^n \quad \bold{G}_{22}^n\bold{V}_2^n]-\Text{rank}[\bold{G}_{22}^n \bold{V}_2^n]$. To do so, we show that for any realization $\mathcal G^n=\{ G_{kj}^n \}_{k,j\in\{1,2\}}$ with the corresponding value $r_1$, and matrices $V_1^n,V_2^n$, we have
$r_1 \stackrel{}{\leq}   \Text{rank}[G_{21}^nV_1^n \quad G_{22}^nV_2^n]-\Text{rank}[G_{22}^n V_2^n]$.

Since $r_1=  \Text{dim}(\Text{span}(\vec s_{ m_1(n) \times 1}|\quad \exists \vec l_{ n \times 1} \quad s.t.\quad  [\vec s^\top \quad \vec 0_{1\times m_2(n)}]=\vec l~^\top [G_{21}^nV_1^n \quad G_{22}^nV_2^n]      ))$,
\begin{equation}
\exists {L}_{r_1\times n} \quad s.t. \quad [{S} \quad 0_{r_1\times m_2(n)}]={L}[G_{21}^nV_1^n \quad G_{22}^n V_2^n],
\end{equation}
for some ${S}_{r_1\times m_1(n)}$, such that $\Text{rank} [{S}]=r_1$.
This means
\begin{equation}
 LG_{22}^nV_2^n=0_{r_1\times m_2(n)}, \quad  LG_{21}^nV_1^n=S,\quad \Text{rank}[{L} G_{21}^nV_1^n]=r_1.\label{x5}
\end{equation}

We now state a claim that will be useful in completing the proof.
\begin{claim} \label{clfrob}
For three matrices $A,B,C$ where the number of columns in $A$ is equal to the number of rows in $B,C$,
\begin{equation}
\Text{rank} [AB \quad  AC] - \Text{rank}[AC] \leq \Text{rank}[B \quad C] - \Text{rank}[C].
\end{equation}
\end{claim}
\begin{proof}
By Frobenius's inequality, for any three matrices $X,Y,Z$ where $XY$, $YZ$, and $XYZ$ are defined,
\begin{equation}
\Text{rank} [XY] + \Text{rank}[YZ] \leq \Text{rank}[XYZ] + \Text{rank} [Y].
\end{equation}
By setting $X=A, Y=[B \quad C], Z=[0 \quad I]^\top$, where $I$ is the identity matrix, the desired result follows.
\end{proof}

Therefore, by setting $A=L,  B= [G_{21}^nV_1^n \quad G_{22}^nV_2^n],  C=G_{22}^nV_2^n$ in Claim \ref{clfrob},   and using (\ref{x5}), we get
\begin{equation}
r_1 -0 \leq   \Text{rank}[G_{21}^nV_1^n \quad G_{22}^nV_2^n]-\Text{rank}[G_{22}^nV_{2}^n],
\end{equation}
which completes the proof.


\section{ Proof of Lemma \ref{main5}} \label{lemm6}
Here we restate Lemma \ref{main5} before proving it.

\begin{lemmarep}{Lemma~\ref{main5}}
Consider a fixed
linear coding strategy $\{  f_1^{(n)},f_2^{(n)} \}$, with corresponding $ \bold{V}_{1}^{n}\stackrel{\Delta}{=} \bold{V}_{11}^{n},\bold{V}_{2}^{n} \stackrel{\Delta}{=} \bold{V}_{12}^{n}$ as defined in (\ref{coding}). For any $i \in \{1,2, \ldots ,n \}$, let $\mathcal  A_i,\mathcal B_i$,  denote the following sets:
\begin{itemize}
\item $\mathcal A_i\triangleq \{ \mathcal G^n| \quad \Text{rank}[G_{21}^{i}V_1^{i} \quad G_{22}^{i}V_2^{i}]=\Text{rank}[G_{21}^{i-1}V_1^{i-1} \quad G_{22}^{i-1}V_2^{i-1}]  \} .  $
\item $\mathcal B_i\triangleq \{ \mathcal G^n|  \quad [\vec v_1(i)^\top \quad \vec 0_{1\times m_2(n)}],[\vec 0_{1\times m_1(n)}\quad \vec v_2(i)^\top]  \in   \Text{rowspan} [G_{21}^{i-1}V_1^{i-1}  \quad  G_{22}^{i-1}V_2^{i-1}] \}  $.
\end{itemize}
Then,
\begin{equation}
\Pr (\bm{\mathcal G}^n \in \cup_{i=1}^{n}(\mathcal A_i\cap \mathcal B_i^c))=0.\nonumber
\end{equation}
\end{lemmarep}

\begin{proof}
Note that due to Union Bound, it is sufficient to show that for any $i \in \{1,2, \ldots ,n \}$,
\begin{equation}
\Pr (\bm{\mathcal G}^n \in \mathcal A_i\cap \mathcal B_i^c)=0.\nonumber
\end{equation}
Consider an arbitrary $i \in \{1,2, \ldots ,n \}$. Due to Total Probability Law, it is sufficient to show that for any channel realization of the first $i-1$ timeslots, denoted by $\mathcal G^{i-1}=\{G_{kj}^{i-1}\}_{j,k\in\{1,2\}}$, we have
\begin{equation}
\label{eq:mainEq1}  \Pr(\bm{\mathcal G}^n \in \mathcal A_i\cap  \mathcal B^c_i | \bm{\mathcal G}^{i-1}=\mathcal G^{i-1})=0.
\end{equation}

Consider an arbitrary channel realization of the first $i-1$ timeslots $\mathcal G^{i-1}=\{G_{kj}^{i-1}\}_{j,k\in\{1,2\}}$, with corresponding matrices $V_1^{i},V_2^{i}$ (which are now deterministic).
Also, suppose that given $\mathcal G^{i-1}$, $ \mathcal  B^c_i$ occurs; since otherwise, the proof would be complete.
On the other hand, assuming $ \mathcal  B^c_i$ occurs, and denoting $\mathcal L = \Text{rowspan} [G_{21}^{i-1}V_1^{i-1}  \quad  G_{22}^{i-1}V_2^{i-1}]$, at least one of the following is true according to the definition of $ \mathcal  B_i$:
\begin{align}
&[\vec v_1(i)^\top \quad \vec 0_{1\times m_2(n)}] \notin   \mathcal L \qquad  \Rightarrow \qquad   \text{Proj}_{\mathcal L^c} [\vec v_1(i)^\top \quad \vec 0_{1\times m_2(n)}] \ne 0 \\
& [\vec 0_{1\times m_1(n)}\quad \vec v_2(i)^\top]  \notin  \mathcal L \qquad  \Rightarrow \qquad   \text{Proj}_{\mathcal L^c} [\vec 0_{1\times m_1(n)}\quad \vec v_2(i)^\top] \ne 0.
\end{align}
Therefore, the $(m_1(n)+m_2(n))\times 2$ matrix 
$[\text{Proj}_{\mathcal L^c} [\vec v_1(i)^\top \quad \vec 0_{1\times m_2(n)}] ^\top   \qquad      \text{Proj}_{\mathcal L^c} [\vec 0_{1\times m_1(n)}\quad \vec v_2(i)^\top]^\top ]$ is non-zero, which means that its null space has dimension strictly lower than $2$.
Hence, we have,
\begin{align*}
& \Pr(\bm{\mathcal G}^n \in \mathcal A_i\cap  \mathcal B^c_i | \bm{\mathcal G}^{i-1}=\mathcal G^{i-1})  \stackrel{(a)}{=}
\Pr(\bm{\mathcal G}^n \in \mathcal A_i| \bm{\mathcal G}^{i-1}=\mathcal G^{i-1})   \\
&  \stackrel{(b)}{=}   \Pr(     \text{Proj}_{\mathcal L^c} [{\bf g}_{21}(i)\vec v_1(i)^\top   \quad {\bf g}_{22}(i)\vec v_2(i)^\top] = 0     | \bm{\mathcal G}^{i-1}=\mathcal G^{i-1})   \\
& \stackrel{(c)}{=}   \Pr(   {\bf g}_{21}(i)  \text{Proj}_{\mathcal L^c} [\vec v_1(i)^\top   \quad 0] +  {\bf g}_{22}(i)\text{Proj}_{\mathcal L^c} [0   \quad \vec v_2(i)^\top] = 0     | \bm{\mathcal G}^{i-1}=\mathcal G^{i-1})   \\
& = \Pr(   [\text{Proj}_{\mathcal L^c} [\vec v_1(i)^\top \quad \vec 0_{1\times m_2(n)}] ^\top   \qquad      \text{Proj}_{\mathcal L^c} [\vec 0_{1\times m_1(n)}\quad \vec v_2(i)^\top]^\top ] \left[\begin{array}{c}{\bf g}_{21}(i)\\{\bf g}_{22}(i)  \end{array}\right]  = 0     | \bm{\mathcal G}^{i-1}=\mathcal G^{i-1}) \\
& = \Pr(  \left[\begin{array}{c}{\bf g}_{21}(i)\\{\bf g}_{22}(i)  \end{array}\right]   \in    \Text{nullspace}[\text{Proj}_{\mathcal L^c} [\vec v_1(i)^\top \quad \vec 0_{1\times m_2(n)}] ^\top   \qquad      \text{Proj}_{\mathcal L^c} [\vec 0_{1\times m_1(n)}\quad \vec v_2(i)^\top]^\top ]    | \bm{\mathcal G}^{i-1}=\mathcal G^{i-1}) \\
& \stackrel{(d)}{=}  0,
\end{align*}
where (a) holds since we assumed that for realization $\mathcal G^{i-1}$, $\mathcal B_i^c$ occurs;
(b) holds according to the definition of $\mathcal A_i$;
(c) holds due to linearity of orthogonal projection;
and (d)  holds since 
the $(m_1(n)+m_2(n))\times 2$ matrix 
$[\text{Proj}_{\mathcal L^c} [\vec v_1(i)^\top \quad \vec 0_{1\times m_2(n)}] ^\top   \qquad      \text{Proj}_{\mathcal L^c} [\vec 0_{1\times m_1(n)}\quad \vec v_2(i)^\top]^\top ]$ is non-zero, which means that its null space, which is a subspace in $\mathbb R^2$, has dimension strictly lower than $2$.
Therefore, the probability that the random vector $ \left[\begin{array}{c}{\bf g}_{21}(i)\\{\bf g}_{22}(i)  \end{array}\right] $ lies in a subspace in $\mathbb R^2$ of strictly lower dimension (than 2) is zero.

\end{proof}

\end{appendices}
\end{document}